\documentclass[letterpaper,twoside,english,final]{IEEEtran}
\usepackage[T1]{fontenc}
\usepackage[latin9]{inputenc}
\usepackage{amsmath}
\usepackage{graphicx}
\usepackage{amssymb}

\makeatletter

\providecommand{\tabularnewline}{\\}

\newtheorem{prop}{Proposition}
\newtheorem{thm}{Theorem}

\newcounter{mytempeqncnt}

\author{Tarik~Ait-Idir,~\IEEEmembership{Member,~IEEE,} Houda~Chafnaji, and Samir~Saoudi,~\IEEEmembership{Senior Member,~IEEE}%
\thanks{The Associate Editor coordinating the review of this paper and approving it for publication is Dr. M. C. Valenti. Manuscript received March 26, 2009; revised December 24, 2009; accepted February 21, 2010. This paper was presented in part at the 19th Annual IEEE Symposium on Personal Indoor and Mobile Radio Communications (PIMRC 2008), Cannes, France, September 2008, and in part at the IEEE Global Communications Conference (Globecom'09), Honolulu, Hawaii, Nov-Dec 2009.}%
\thanks{T. Ait-Idir and H. Chafnaji are with the Communication Systems Department, INPT, Madinat Al-Irfane, Rabat, Morocco. They are also with Institut Telecom / Telecom Bretegne/LabSticc, Brest, France (email: aitidir@ieee.org).}%
\thanks{ S. Saoudi is with Institut Telecom / Telecom Bretegne/LabSticc, Brest, France. He is also with Université Européenne de Bretagne.}}%
\vspace{1cm}

\makeatother

\usepackage{babel}

\begin{document}

\title{Turbo Packet Combining for Broadband Space--Time BICM Hybrid--ARQ
Systems with Co--Channel Interference}
\maketitle
\begin{abstract}
In this paper, efficient turbo packet combining for single carrier
(SC) broadband multiple-input--multiple-output (MIMO) hybrid--automatic
repeat request (ARQ) transmission with unknown co-channel interference
(CCI) is studied. We propose a new frequency domain soft minimum mean
square error (MMSE)-based signal level combining technique where received
signals and channel frequency responses (CFR)s corresponding to all
retransmissions are used to decode the data packet. We provide a recursive
implementation algorithm for the introduced scheme, and show that
both its computational complexity and memory requirements are quite
insensitive to the ARQ delay, i.e., maximum number of ARQ rounds.
Furthermore, we analyze the asymptotic performance, and show that
under a sum-rank condition on the CCI MIMO ARQ channel, the proposed
packet combining scheme is not interference-limited. Simulation results
are provided to demonstrate the gains offered by the proposed technique. \end{abstract}

\section{Introduction\label{sec:Introduction}}

\PARstart{S}{pace--time--bit-interleaved} coded modulation (ST--BICM)
with iterative decoding is an attractive signaling scheme that offers
high spectral efficiencies over multiple-input--multiple-output (MIMO)-intersymbol
interference (ISI) channels \cite{Ariyavisitakul_TCOM00,Tonello,vandendorpe_SigProc_04,Visoz-groupMMSE-journal,Ait-idir_TVT}.
To combat ISI in single carrier (SC) broadband ST--BICM transmission,
frequency domain equalization, initially introduced for single antenna
systems \cite{Sari_Globecom94,Tuchler_Allerton00,Falconer_CommMag02,Vitetta_TCOM05},
has been proposed using iterative (turbo) processing \cite{Visoz_TCOM06}.
It is a receiver scheme that allows high ISI cancellation capability
at an affordable complexity cost. In practical systems, unknown co-channel
interference (CCI) caused by other transmitters (distant users and/or
neighboring cells) who simultaneously use the same radio resource
can dramatically degrade the link performance. This limitation can
be overcome by using the so-called hybrid--automatic repeat request
(ARQ) protocols, where channel coding is combined with ARQ \cite{Costello_book,Costello_Hagenauer_IT98}.
In hybrid--ARQ, erroneous data packets are kept in the receiver and
used to detect/decode the retransmitted frame \cite{SindhuTCOM77,Benelli_TCOM85,Wicker_TCOM91,Yli-Juuti_CommLett98,Liinaharja_CommLett99,Harvey_TCOM_94,Kallel_TCOM90}.
This technique is often referred to as \emph{{}``packet combining''}.
Practical packet combining schemes have been addressed in \cite{McTiffin_et_al94}.
In \cite{Caire_IT01}, an elegant information-theoretic framework
has been introduced to analyze the throughput and delay of hybrid--ARQ
under random user behavior. Interestingly, the authors have shown
that hybrid--ARQ systems are not interference limited, i.e., arbitrarily
high throughput can be achieved by simply increasing the transmit
power of all users even when multi-user detection (MUD) techniques
are not used at the receiver. Motivated by the above considerations,
we investigate efficient low-complexity turbo frequency domain reception
techniques for SC broadband ST--BICM signaling with hybrid--ARQ operating
over CCI-limited MIMO channels. 

The powerful diversity--multiplexing tradeoff tool, initially introduced
by Zheng and Tse for coherent delay-limited, i.e., quasi-static, MIMO
channels \cite{Zheng-Tse}, has been elegantly extended by El Gamal
\emph{et al.} to MIMO ARQ channels with flat fading, and referred
to as diversity--multiplexing--delay tradeoff \cite{DMD_Gamal_IT_2006}.
The authors have proved that the ARQ delay, i.e., maximum number of
ARQ protocol rounds, improves the outage probability %
\footnote{In non-ergodic, i.e., block fading quasi-static channels, the outage
probability is a meaningful measure that provides a lower bound on
the block error probability. It is defined as the probability that
the mutual information, as a function of the channel realization and
the average signal-to-noise ratio (SNR), is below the transmission
rate \cite{Tse_Viswanath_Book}. %
} performance for large classes of MIMO ARQ channels \cite{DMD_Gamal_IT_2006}.
In particular, they have demonstrated that the diversity order can
be increased due to ARQ even when the MIMO ARQ channel is long-term
static, i.e., the MIMO channel is random but fixed for all ARQ rounds.
The diversity--multiplexing--delay tradeoff has then been characterized
in the case of block-fading MIMO ARQ channels, i.e., multiple fading
blocks are allowed within the same ARQ round \cite{Chuanng_DMD_BlockFading_IT08}.
In \cite{Ait-Idir_MIMOISIARQ_TCOM08}, the outage probability of MIMO-ISI
ARQ channels has been evaluated under the assumptions of short-term
static channel dynamic %
\footnote{In the case of short-term static dynamic, the ARQ channel realizations
are independent from round to round. This dynamic applies to slow
ARQ protocols where the delay between two rounds is larger than the
channel coherence time. %
}, and Chase-type ARQ \cite{ChaseTComm1985}, i.e., the data packet
is entirely retransmitted. It has been shown that, as in the flat
fading case, ARQ presents an important source of diversity, but its
influence becomes only minimal when the ARQ delay is increased. This
observation suggests that the design of practical packet combining
schemes should target a high diversity order for early ARQ rounds.
Supplementary retransmissions are then used to correct rare erroneous
data packets, when they occur. 

More recently, packet combining for MIMO ARQ systems has been investigated
(e.g. \cite{H_Zheng_et_al_PIMRC02,Dabak_ICC03,ZhizhongDing_RiceICC03,Hottinen_et_al_ISSPA03,Samra_Ding04,Koike_ICC04,Ding_MIMOARQ_Sphere_SigProc,Krishnaswamy_VTCfall06,Cioffi_TWC_Feb09}).
Turbo combining techniques, where decoding is iteratively performed
through the exchange of soft information between the soft-input--soft-output
(SISO) packet combiner and the SISO decoder, have been proposed for
the MIMO-ISI ARQ channel using unconditional minimum mean square error
(MMSE)-aided combining \cite{Ait-Idir_WCNC_08,Ait-Idir_MIMOISIARQ_TCOM08}.
These approaches have then been extended to broadband MIMO code division
multiple access (CDMA) systems with ARQ \cite{Chafnaji_PIMRC08}.
Time domain turbo packet combining for CCI-limited MIMO-ISI ARQ channels
has been introduced in \cite{Ait-Idir_PIMRC08}. 

In this paper, we investigate efficient turbo receiver techniques
for SC ST--BICM transmission with Chase-type ARQ over broadband MIMO
channel with \emph{unknown} CCI. We introduce a frequency domain MMSE-based
turbo packet combining scheme, where all ARQ rounds are used to decode
the data packet. By using an identical cyclic prefix (CP) word for
multiple retransmissions of a symbol block, we perform transmission
combining at the signal level. The frequency domain soft MMSE packet
combiner performs soft ISI cancellation and retransmission combining
in the presence of unknown CCI jointly over all received signal blocks.
We also provide an efficient recursive implementation algorithm for
the proposed scheme, and show that both the computational load and
memory requirements are quite insensitive to the ARQ delay. The complexity
order is only cubic in terms of the number of transmit antennas. Received
signals and channel frequency responses (CFR)s corresponding to all
ARQ rounds are used without being required to be stored in the receiver.
We analyze the asymptotic performance of the proposed combining scheme.
Interestingly, we show that under a rank-condition on the MIMO ARQ
channel corresponding to unknown CCI, the proposed combining scheme
is not interference-limited, i.e.,unknown CCI can be completely removed.
Finally, we provide numerical simulation results for some scenarios
to validate our findings. 

The remainder of the paper is organized as follows. In Section \ref{sec:SC-MIMO-ARQscheme_and_CommModels}
we describe the ARQ system under consideration, along with the communication
model in the presence of unknown CCI. In Section \ref{sec:Turbo-Packet-Combining},
we present the frequency domain turbo packet combining scheme we propose
in this paper, and analyze both its complexity and memory requirements.
In Section \ref{sec:Performance-Evaluation}, we carry out the asymptotic
performance analysis, and provide representative numerical results
that demonstrate the gains achieved by the proposed scheme. Finally,
we point out conclusions in Section \ref{sec:Conclusion}.

\emph{Notation:}
\begin{itemize}
\item Superscripts $^{\star}$, $^{\top}$, and $^{H}$ denote conjugate,
transpose, and Hermitian transpose, respectively. $\mathbb{{E}}\left[.\right]$
is the mathematical expectation of the argument $\left(.\right)$. 
\item Let $\mathbf{X}$ be a square matrix, $\mathrm{diag}\left\{ \mathbf{X}\right\} $
denotes the row vector corresponding to the diagonal of $\mathbf{X}$,
and $\mathrm{tr}\left\{ \mathbf{X}\right\} $ denotes the trace of
$\mathbf{X}$. When $\mathbf{X}_{1},\cdots,\mathbf{X}_{M}\in\mathbb{C}^{N\times Q}$,
$\mathrm{diag}\left\{ \mathbf{X}_{1},\cdots,\mathbf{X}_{M}\right\} $
denotes the $MN\times MQ$ matrix whose diagonal blocks are $\mathbf{X}_{1},\cdots,\mathbf{X}_{M}$.
$\mathrm{diag}\left\{ \mathbf{x}\right\} $ is the $N\times N$ diagonal
matrix whose diagonal entries are the elements of the complex vector
$\mathbf{x}\in\mathbb{C}^{N}$. $\left(\mathbf{X}\right)_{m,m}$ denotes
the $m$th diagonal entry of matrix $\mathbf{X}$.
\item $\mathbf{I}_{N}$ is the $N\times N$ identity matrix, and $\mathbf{0}_{N\times Q}$
denotes an all zero $N\times Q$ matrix. For $i=0,\cdots,T-1$, $\mathbf{E}_{i,N}$
is a $N\times NT$ zero matrix where the $i$th $N\times N$ block
is equal to $\mathbf{I}_{N}$.
\item Operator $\otimes$ denotes the Kronecker product, and $\delta_{m,n}$
is the Kronecker symbol, i.e., $\delta_{m,n}=1$ for $m=n$ and $\delta_{m,n}=0$
for $m\neq n$.
\item For each sequence of matrices $\mathbf{X}_{0},\cdots,\mathbf{X}_{T-1}$
(respectively, scalars $x_{0},\cdots,x_{T-1}$), $\tilde{\mathbf{X}}\triangleq\frac{1}{T}\sum_{i=0}^{T-1}\mathbf{X}_{i}$
denotes its time average (respectively, $\tilde{x}\triangleq\frac{1}{T}\sum_{i=0}^{T-1}x_{i}$).
\item $\mathbf{U}_{T}$ is a $T\times T$ unitary matrix whose $\left(m,n\right)$th
element is $\left(\mathbf{U}_{T}\right)_{m,n}=\frac{1}{\sqrt{T}}\exp\left\{ -j\frac{2\pi mn}{T}\right\} $
for $m,\, n=0,\cdots,T-1$, where $j=\sqrt{-1}$. $\mathbf{U}_{T,N}$
is $TN\times TN$ defined as $\mathbf{U}_{T,N}\triangleq\mathbf{U}_{T}\otimes\mathbf{I}_{N}$.
\item For each vector $\mathbf{x}\in\mathbb{C}^{Q}$, $\mathbf{x}_{f}$
denotes the discrete Fourier Transform (DFT) of $\mathbf{x}$, i.e.,
$\mathbf{x}_{f}=\mathbf{U}_{Q}\,\mathbf{x}$.
\item The acronym i.i.d. means {}``independent and identically distributed''. 
\end{itemize}

\section{ARQ System Model \label{sec:SC-MIMO-ARQscheme_and_CommModels}}

\subsection{SC--MIMO ARQ Transmission Scheme}

We consider an SC multi-antenna-aided transmission scheme where the
transmitter and the receiver are equipped with $N_{T}$ transmit (index
$t=1,\cdots,N_{T}$) and $N_{R}$ receive (index $r=1,\cdots,N_{R}$)
antennas, respectively. The MIMO channel is frequency selective and
is composed of $L$ symbol-spaced taps (index $l=0,\cdots,L-1$).
The energy of each tap $l$ is denoted $\sigma_{l}^{2}$, and the
total energy is normalized to one, i.e., $\sum_{l=0}^{L-1}\sigma_{l}^{2}=1.$

Each information block is initially encoded then interleaved with
the aid of a semi-random interleaver $\Pi$. The resulting frame is
\emph{serial to parallel} converted and mapped over the elements of
the constellation set $\mathcal{S}$ to produce symbol matrix $\mathbf{S}\in\mathcal{S}^{N_{T}\times T}$,
where $T$ is the number of channel use (c.u). A CP word, whose length
is $T_{CP}\geq L-1$, is then appended to $\mathbf{S}$, thereby yielding
matrix $\mathbf{S}'\in\mathcal{S}^{N_{T}\times\left(T+T_{CP}\right)}$.
This allows the prevention of inter-block interference (IBI) and the
exploitation of the multipath diversity of the MIMO broadband channel.
We suppose that no channel state information (CSI) is available at
the transmitter and assume infinitely deep interleaving. Therefore,
transmitted symbols are independent and have equal transmit power,
i.e., 

\begin{equation}
\mathbb{{E}}\left[s_{t,i}s_{t',i'}^{\star}\right]=\delta_{t-t',i-i'}.\label{eq:Symb_independency}\end{equation}

At the upper layer, an ARQ protocol is used to help correct erroneous
frames. An acknowledgment message is generated after the decoding
of each information block. Therefore, when the decoding is successful,
the receiver sends back a positive acknowledgment (ACK) to the transmitter,
while the feedback of a negative acknowledgment (NACK) indicates that
the decoding outcome is erroneous. Let $K$ denote the ARQ delay,
and $k=1,\cdots,K$ denote the ARQ round index. When the transmitter
receives an ACK feedback, it stops the transmission of the current
block and moves on to the next information block. Reception of a NACK
message incurs supplementary ARQ rounds until the packet is correctly
decoded or the ARQ delay $K$ is reached. We focus on Chase-type ARQ,
i.e., the symbol matrix $\mathbf{S}'$ is completely retransmitted.
In addition, we suppose perfect packet error detection, and assume
that the one bit ACK/NACK feedback is error-free.%
\begin{figure*}[t]
\noindent \begin{centering}
\includegraphics[scale=0.55]{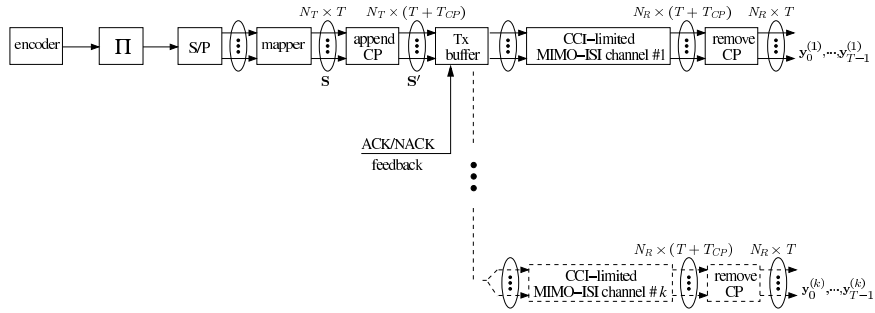}
\par\end{centering}

\caption{\label{fig:STBICM_ARQ} SC--MIMO ARQ communication scheme at ARQ round
$k$.}

\end{figure*}

\subsection{Communication Model in the Presence of Unknown CCI}

The broadband MIMO ARQ channel is assumed to be short-term static
fading, i.e., the channel independently changes from round to round.
Note that this channel dynamic applies to slow ARQ protocols where
the delay between two consecutive ARQ rounds is larger than the channel
coherence time. It also applies to orthogonal frequency division multiplexing
(OFDM) systems where frequency hopping is used to mitigate ISI. Let
$\mathbf{H}_{0}^{\left(k\right)},\cdots,\mathbf{H}_{L-1}^{\left(k\right)}\in\mathbb{C}^{N_{R}\times N_{T}}$
denote channel matrices at the $k$th ARQ round, and whose entries
are i.i.d. zero-mean circularly symmetric Gaussian, i.e., $h_{r,t,l}^{\left(k\right)}\sim\mathcal{CN}\left(0,\sigma_{l}^{2}\right)$,
where $h_{r,t,l}^{\left(k\right)}$ denotes the fading channel corresponding
to path $l$ and connecting the $t$th transmit and the $r$th receive
antennas at the $k$th ARQ round. Therefore, the channel energy at
each receive antenna $r$ is \begin{equation}
{\textstyle {\displaystyle \sum_{l=0}^{L-1}}{\displaystyle \sum_{t=1}^{N_{T}}}}\,\mathbb{{E}}\left[\left|h_{r,t,l}^{\left(k\right)}\right|^{2}\right]=N_{T}.\label{eq:Ch_energy_Rx_Ant}\end{equation}
The channel profile, i.e., power distribution $\sigma_{0}^{2},\cdots,\sigma_{L-1}^{2}$
and number of taps $L$, is supposed to be identical for at least
$K$ consecutive rounds. This is a reasonable assumption because the
channel profile dynamic mainly depends on the shadowing effect. 

Transmitted data blocks are corrupted by an \emph{unknown} CCI signal
caused by a co--channel transmission that uses $N'_{T}$ transmit
antennas (index $t'=1,\cdots,N'_{T}$) and $T$ c.u. The link between
the interferer transmitter and the receiver is composed of $L'$ taps,
where the channel matrix of each tap $l'=0,\cdots,L'-1$ at round
$k$ is $\mathbf{H}_{l'}^{\mathrm{CCI}^{\left(k\right)}}\in\mathbb{C}^{N_{R}\times N'_{T}}$
and its energy is $\sigma_{u_{l'}}^{2}$ %
\footnote{The ARQ processes corresponding to the desired user and the interferer
are not necessarily synchronized. Therefore, the round index $k$
appearing in the CCI channel matrices only refers to the index of
a realization of the interferer channel at ARQ round $k$. The same
remark holds for CCI symbols in (\ref{eq:Signal_Model}). Also, note
that $\sum_{l'=0}^{L'-1}\sigma_{u_{l'}}^{2}\neq1$ in order to account
for the path-loss between the interferer and the receiver. %
}. We suppose that the receiver has no knowledge either about the interferer
CSI or about its channel profile and number of transmit antennas (i,.e.,
parameters $N'_{T}$, $L'$, $\mathbf{H}_{l'}^{\mathrm{CCI}^{\left(k\right)}}$,
and $\sigma_{u_{l'}}^{2}$ $\forall l'$ are completely unknown at
the receiver). As the desired user, the interferer employs a CP-aided
transmission strategy. Its transmitted symbols $s_{t',i}^{\mathrm{CCI}^{\left(k\right)}}$at
each round $k$ verify the independence/energy-normalization condition
(\ref{eq:Symb_independency}) as useful symbols. Therefore, the signal-to-interference
ratio (SIR) at each receive antenna is given as \begin{equation}
\mathrm{SIR}=\frac{N_{T}}{N'_{T}\sum_{l'=0}^{L'-1}\sigma_{u_{l'}}^{2}}.\label{eq:SIR}\end{equation}
We assume perfect frame synchronization between the interferer and
the desired user. They can differ in terms of the CP word length,
which depends on the delay of the multipath channel, but are synchronized
in terms of the useful symbol frames. Under this assumption, CP deletion
yields the following baseband received $N_{R}\times1$ signal at round
$k$ and channel use $i$, \begin{equation}
\mathbf{y}_{i}^{\left(k\right)}=\sum_{l=0}^{L-1}\mathbf{H}_{l}^{\left(k\right)}\mathbf{s}_{\left(i-l\right)\,\mathrm{mod}\, T}+\underbrace{\sum_{l'=0}^{L'-1}\mathbf{H}_{l'}^{\mathrm{CCI}^{\left(k\right)}}\mathbf{s}_{\left(i-l'\right)\,\mathrm{mod}\, T}^{\mathrm{CCI}^{\left(k\right)}}+\mathbf{n}_{i}^{\left(k\right)}}_{\mathbf{w}_{i}^{\left(k\right)}\,=\,\mathrm{CCI+noise}},\label{eq:Signal_Model}\end{equation}
where $\mathbf{n}_{i}^{\left(k\right)}\sim\mathcal{CN}\left(\mathbf{0}_{N_{R}\times1},\sigma^{2}\mathbf{I}_{N_{R}}\right)$
denotes the receiver thermal noise. The SC--MIMO ARQ communication
scheme at round $k$ is depicted in Fig.  \ref{fig:STBICM_ARQ}. In
the following, we assume perfect channel estimation at each ARQ round
$k$ (i.e., $\mathbf{H}_{l}^{\left(k\right)}$ $\forall l$ are perfectly
known) while CCI channel matrices $\mathbf{H}_{l'}^{\mathrm{CCI}^{\left(k\right)}}$
$\forall l',\, k$ are completely unknown at the receiver side.

\subsubsection{Single-Round Communication Model}

To derive the block communication model corresponding to ARQ round
$k$, we consider the following block signal vector, \begin{equation}
\mathbf{y}^{\left(k\right)}\triangleq\left[\mathbf{y}_{0}^{\left(k\right)^{\top}},\cdots,\mathbf{y}_{T-1}^{\left(k\right)^{\top}}\right]^{\top}\in\mathbb{C}^{N_{R}T},\label{eq:Rx_Signal_all_cu_round_k}\end{equation}
that groups signals corresponding to the entire symbol frame. Vector
$\mathbf{y}^{\left(k\right)}$ can be expressed as, \begin{equation}
\mathbf{y}^{\left(k\right)}=\mathbf{H}^{\left(k\right)}\mathbf{s}+\mathbf{w}^{\left(k\right)},\label{eq:Block_Comm_Model_singleRound}\end{equation}
where \begin{equation}
\mathbf{s}\triangleq\left[\mathbf{s}_{0}^{\top},\cdots,\mathbf{s}_{T-1}^{\top}\right]^{\top}\in\mathcal{S}^{TN_{T}},\label{eq:Symbols_BlockVector}\end{equation}
\begin{equation}
\mathbf{w}^{\left(k\right)}\triangleq\left[\mathbf{w}_{0}^{\left(k\right)^{\top}},\cdots,\mathbf{w}_{T-1}^{\left(k\right)^{\top}}\right]^{\top}\in\mathbb{C}^{N_{R}T},\label{eq:CCI_blockVector}\end{equation}
\begin{equation}
\mathbf{H}^{\left(k\right)}\triangleq\left[\begin{array}{cccc}
\mathbf{H}_{0}^{\left(k\right)} & \mathbf{0}_{N_{R}\times N_{T}} & \cdots & \mathbf{0}_{N_{R}\times N_{T}}\\
\vdots & \mathbf{H}_{0}^{\left(k\right)} &  & \vdots\\
\mathbf{H}_{L-1}^{\left(k\right)} & \vdots &  & \vdots\\
\mathbf{0}_{N_{R}\times N_{T}} & \mathbf{H}_{L-1}^{\left(k\right)} &  & \mathbf{0}_{N_{R}\times N_{T}}\\
\vdots & \mathbf{0}_{N_{R}\times N_{T}} &  & \mathbf{H}_{L-1}^{\left(k\right)}\\
\vdots & \vdots &  & \vdots\\
\mathbf{0}_{N_{R}\times N_{T}} & \mathbf{0}_{N_{R}\times N_{T}} & \cdots & \mathbf{H}_{0}^{\left(k\right)}\end{array}\right]_{N_{R}T\times N_{T}T}\label{eq:Block_circulant_Mat_singleRound}\end{equation}
is a block circulant matrix that can be block diagonalized in a Fourier
basis as \begin{equation}
\mathbf{H}^{\left(k\right)}=\mathbf{U}_{T,N_{R}}^{H}\mathbf{\Lambda}^{\left(k\right)}\mathbf{U}_{T,N_{T}},\label{eq:Block_diagnalization_singleRound}\end{equation}
where \begin{equation}
\mathbf{\Lambda}^{\left(k\right)}\triangleq\mathrm{diag}\left\{ \mathbf{\Lambda}_{0}^{\left(k\right)},\cdots,\mathbf{\Lambda}_{T-1}^{\left(k\right)}\right\} \in\mathbb{C}^{N_{R}T\times N_{T}T}.\label{eq:Block_diag_Mat_SingleRound}\end{equation}
Exploiting (\ref{eq:Block_diagnalization_singleRound}) and the block
circulant structure of $\mathbf{H}^{\left(k\right)}$, we get \begin{equation}
\mathbf{\Lambda}_{i}^{\left(k\right)}=\sum_{l=0}^{L-1}\mathbf{H}_{l}^{\left(k\right)}\exp\left\{ -j\frac{2\pi il}{T}\right\} .\label{eq:CSI_DFT_oneRound}\end{equation}
Applying the DFT $\mathbf{U}_{T,N_{R}}$ on signal vector $\mathbf{y}^{\left(k\right)}$
yields the single-round frequency domain communication model \begin{equation}
\mathbf{y}_{f}^{\left(k\right)}=\mathbf{\Lambda}^{\left(k\right)}\mathbf{s}_{f}+\mathbf{w}_{f}^{\left(k\right)},\label{eq:FD_CommModel_oneARQround}\end{equation}
where $\mathbf{y}_{f}^{\left(k\right)}$, $\mathbf{s}_{f}$, and $\mathbf{w}_{f}^{\left(k\right)}$
denote the DFT of $\mathbf{y}^{\left(k\right)}$, $\mathbf{s}$, and
$\mathbf{w}^{\left(k\right)}$, respectively.

\subsubsection{Multi-Round Communication Model}

Let us suppose that received signals and channel matrices corresponding
to ARQ rounds $1,\cdots,k$ are available at the receiver. First,
we introduce the signal vector notation \begin{equation}
\underline{\mathbf{y}}_{i}^{\left(k\right)}\triangleq\left[\mathbf{y}_{i}^{\left(1\right)^{\top}},\cdots,\mathbf{y}_{i}^{\left(k\right)^{\top}}\right]^{\top}\in\mathbb{C}^{kN_{R}},\label{eq:multipleRx_signal}\end{equation}
where received signals corresponding to multiple ARQ rounds are grouped
in such a way to construct $kN_{R}$ virtual receive antennas. Similarly,
we define, \begin{equation}
\underline{\mathbf{H}}_{l}^{\left(k\right)}\triangleq\left[\mathbf{H}_{l}^{\left(1\right)^{\top}},\cdots,\mathbf{H}_{l}^{\left(k\right)^{\top}}\right]^{\top}\in\mathbb{C}^{kN_{R}\times N_{T}},\label{eq:mutli_round_CSI}\end{equation}
\begin{equation}
\underline{\mathbf{w}}_{i}^{\left(k\right)}\triangleq\left[\mathbf{w}_{i}^{\left(1\right)^{\top}},\cdots,\mathbf{w}_{i}^{\left(k\right)^{\top}}\right]^{\top}\in\mathbb{C}^{kN_{R}}.\label{eq:multipleRounds_CCInoise}\end{equation}
The block signal vector that serves for jointly performing, at ARQ
round $k$, packet combining and equalization in the presence of CCI
is constructed similarly to (\ref{eq:Rx_Signal_all_cu_round_k}),
\begin{equation}
\underline{\mathbf{y}}^{\left(k\right)}\triangleq\left[\underline{\mathbf{y}}_{0}^{\left(k\right)^{\top}},\cdots,\underline{\mathbf{y}}_{T-1}^{\left(k\right)^{\top}}\right]^{\top}\in\mathbb{C}^{kN_{R}T},\label{eq:Sig_combining}\end{equation}
and can be expressed as, \begin{equation}
\underline{\mathbf{y}}^{\left(k\right)}=\underline{\mathbf{H}}^{\left(k\right)}\mathbf{s}+\underline{\mathbf{w}}^{\left(k\right)},\label{eq:Block_Comm_Model_Combining}\end{equation}
where \begin{equation}
\underline{\mathbf{w}}^{\left(k\right)}\triangleq\left[\underline{\mathbf{w}}_{0}^{\left(k\right)^{\top}},\cdots,\underline{\mathbf{w}}_{T-1}^{\left(k\right)^{\top}}\right]^{\top}\in\mathbb{C}^{kN_{R}T}.\label{eq:CCI_blockVector}\end{equation}
Matrix $\underline{\mathbf{H}}^{\left(k\right)}$ has the same structure
as (\ref{eq:Block_circulant_Mat_singleRound}), where its first $TkN_{R}\times N_{T}$
block column is equal to\begin{equation}
\left[\underline{\mathbf{H}}_{0}^{\left(k\right)^{\top}},\cdots,\underline{\mathbf{H}}_{L-1}^{\left(k\right)^{\top}},\mathbf{0}_{N_{T}\times\left(T-L\right)kN_{R}}\right]^{\top}.\label{eq:Block_circ_mat_1st_col}\end{equation}
$\underline{\mathbf{H}}^{\left(k\right)}$ can be factorized similarly
to (\ref{eq:Block_diagnalization_singleRound}) as,\begin{equation}
\underline{\mathbf{H}}^{\left(k\right)}=\mathbf{U}_{T,kN_{R}}^{H}\underline{\mathbf{\Lambda}}^{\left(k\right)}\mathbf{U}_{T,N_{T}},\label{eq:Block_diagnalization}\end{equation}
where\begin{equation}
\underline{\mathbf{\Lambda}}^{\left(k\right)}\triangleq\mathrm{diag}\left\{ \underline{\mathbf{\Lambda}}_{0}^{\left(k\right)},\cdots,\underline{\mathbf{\Lambda}}_{T-1}^{\left(k\right)}\right\} \in\mathbb{C}^{kN_{R}T\times N_{T}T},\label{eq:Block_diag_Mat}\end{equation}
\begin{equation}
\underline{\mathbf{\Lambda}}_{i}^{\left(k\right)}\triangleq\left[\mathbf{\Lambda}_{i}^{\left(1\right)^{\top}},\cdots,\mathbf{\Lambda}_{i}^{\left(k\right)^{\top}}\right]^{\top}\in\mathbb{C}^{kN_{R}\times N_{T}},\label{eq:CSI_DFT_allRounds}\end{equation}
and matrices $\mathbf{\Lambda}_{i}^{\left(k'\right)}$, $k'=1,\cdots,k$,
are given by (\ref{eq:CSI_DFT_oneRound}). The multi-round frequency
domain communication model at ARQ round $k$ is then expressed as,
\begin{equation}
\underline{\mathbf{y}}_{f}^{\left(k\right)}=\underline{\mathbf{\Lambda}}^{\left(k\right)}\mathbf{s}_{f}+\underline{\mathcal{\mathbf{w}}}_{f}^{\left(k\right)},\label{eq:FD_BlockCommModel_round_k}\end{equation}
where $\underline{\mathbf{y}}_{f}^{\left(k\right)}$ and $\underline{\mathcal{\mathbf{w}}}_{f}^{\left(k\right)}$
denote the DFT of $\underline{\mathbf{y}}^{\left(k\right)}$ and $\underline{\mathbf{w}}^{\left(k\right)}$,
respectively.%
\begin{figure}[t]
\noindent \begin{centering}
\includegraphics[scale=0.49]{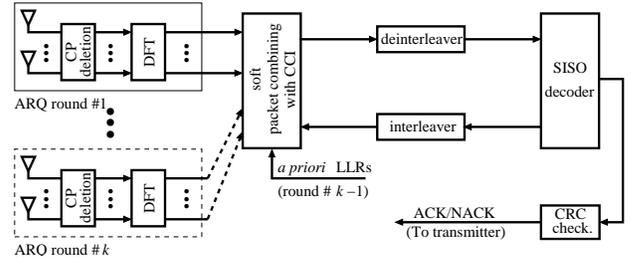}
\par\end{centering}

\caption{\label{fig:receiver_scheme} Block diagram of the receiver scheme
at ARQ round $k$}

\end{figure}

\section{Frequency Domain Turbo Packet Combining in the Presence of Unknown
CCI\label{sec:Turbo-Packet-Combining}}

\subsection{General Description}

At each ARQ round, the decoding of a data packet is performed by iteratively
exchanging soft information in the form of log-likelihood ratio (LLR)
values between the \emph{soft packet combiner}, i.e., the joint transmission
combining and equalization unit, and the soft-input--soft-output (SISO)
decoder. Let us suppose that, at ARQ round $k$, all received signals
and channel matrices corresponding to previous rounds $k-1,\cdots,1$
are available at the receiver. Note that this assumption could not
be feasible in practice since the receiver will require a huge memory.
In Subsection \ref{sub:Implementation}, we show that the proposed
turbo packet combining algorithm requires little memory while it uses
signals and CSIs corresponding to all ARQ rounds $1,\cdots,k$. The
block diagram of the frequency domain turbo packet combining receiver
at ARQ round $k$ is depicted in Fig. \ref{fig:receiver_scheme}.

First, the multiple ARQ rounds frequency domain block signal vector
$\underline{\mathbf{y}}_{f}^{\left(k\right)}\triangleq\mathbf{U}_{T,kN_{R}}\underline{\mathbf{y}}^{\left(k\right)}$
and CFR $\underline{\mathbf{\Lambda}}^{\left(k\right)}$ are constructed.
Second, the soft packet combiner estimates the covariance of unknown
CCI plus noise, then computes the multi-transmission MMSE filter that
takes into account both co-antenna interference (CAI) and ISI while
suppressing unknown CCI. These two elements are then used with \emph{a
priori} information to compute extrinsic LLRs corresponding to coded
and interleaved bits. The generated soft information is transferred
to the SISO decoder to compute \emph{a posteriori} LLRs about both
coded and useful bits. Only extrinsic information is fed back to the
soft packet combiner to help perform transmission combining and equalization
in the next turbo iteration. The iterative soft packet combining and
decoding process is stopped after a preset number of turbo iterations
and decision about the data packet is performed. The ACK/NACK message
is then sent back to the transmitter depending on the decoding outcome.
Note that during the first iteration\emph{ a priori} LLR values are
the output of the SISO decoder obtained at the last iteration of previous
round $k-1$.

\subsection{Properties of CCI plus Noise Covariance\label{sub:MultipleRoundCov_Properties} }

In this subsection, we focus on covariance properties of CCI plus
noise present in both the single-round and multi-round communication
models given by (\ref{eq:Block_Comm_Model_singleRound}) and (\ref{eq:Block_Comm_Model_Combining}),
respectively. These properties present an important ingredient in
the turbo packet combining algorithm we introduce in Subsection \ref{sub:ProposedScheme}.

Let \emph{$\mathbf{\Theta}_{k}$ }denote the covariance of CCI plus
noise $\mathbf{w}_{i}^{\left(k\right)}$ present in received signal
(\ref{eq:Signal_Model}) at round $k$, \begin{equation}
\mathbf{\mathbf{\Theta}}_{k}\triangleq\mathbb{{E}}\left[\mathbf{w}_{i}^{\left(k\right)}\mathbf{w}_{i}^{\left(k\right)^{H}}\right]\in\mathbb{C}^{N_{R}\times N_{R}}.\label{eq:CCI_noise_cov_Rx_Signal}\end{equation}
Let us group covariance matrices corresponding to rounds $1,\cdots,k$
in the block diagonal matrix \emph{\begin{equation}
\boldsymbol{{\Xi}}_{k}\triangleq\mathrm{diag}\left\{ \mathbf{\Theta}_{1},\cdots,\mathbf{\Theta}_{k}\right\} \mathbb{\in C}^{kN_{R}\times kN_{R}}.\label{eq:Concat_Cov_one_round}\end{equation}
}
\begin{prop}
\label{pro:MultiRound_Cov}The covariance $\underline{\boldsymbol{{\Xi}}}_{k}\triangleq\mathbb{{E}}\left[\underline{\mathbf{w}}^{\left(k\right)}\underline{\mathbf{w}}^{\left(k\right)^{H}}\right]$
of the CCI plus noise block vector $\underline{\mathbf{w}}^{\left(k\right)}$
present in the multi-round communication model (\ref{eq:Block_Comm_Model_Combining})
after $k$ rounds is expressed as \emph{\begin{equation}
\underline{\boldsymbol{{\Xi}}}_{k}=\mathbf{I}_{T}\otimes\boldsymbol{{\Xi}}_{k}\,\,\mathbb{\in C}^{TkN_{R}\times TkN_{R}}.\label{eq:Cov_AllRounds}\end{equation}
}\end{prop}
\begin{proof}
The expression in (\ref{eq:Cov_AllRounds}) is easily obtained by
calculating the mathematical expectation of $\underline{\mathbf{w}}^{\left(k\right)}\underline{\mathbf{w}}^{\left(k\right)^{H}}$.
In the derivation, we only exploit the independence between the entries
of $\mathbf{H}_{l}^{\mathrm{CCI}^{\left(k\right)}}$ and $\mathbf{H}_{l'}^{\mathrm{CCI}^{\left(k'\right)}}$
$\forall\, l,l',k,$ and $k'$ (i.e., short-term static block fading
dynamic of the CCI MIMO ARQ channel), and the fact that CCI symbols
satisfy (\ref{eq:Symb_independency}). No assumption on the structure
of the CCI block matrix is used. A detailed proof of (\ref{eq:Cov_AllRounds})
in the case of sliding-window aided time-domain detection can be found
in \cite[Subsection III.C]{Ait-Idir_PIMRC08}.\end{proof}
\begin{prop}
\label{pro:Cov_singleRound}The covariance $\underline{\boldsymbol{\mathbf{\Theta}}}_{k}\triangleq\mathbb{{E}}\left[\mathbf{w}^{\left(k\right)}\mathbf{w}^{\left(k\right)^{H}}\right]$
of the single-round CCI plus noise block vector $\mathbf{w}^{\left(k\right)}$
at ARQ round $k$ is \begin{equation}
\underline{\boldsymbol{\mathbf{\Theta}}}_{k}=\mathbf{I}_{T}\otimes\mathbf{\Theta}_{k}.\label{eq:Round_k_CCIplusNoise_Cov}\end{equation}
\end{prop}
\begin{proof}
The proof follows by simply invoking Proposition \ref{pro:MultiRound_Cov}
for one round. \end{proof}
\begin{prop}
\label{pro:Cov_FD_CCI}Covariance matrices of frequency domain CCI
plus noise vectors\emph{ $\underline{\mathcal{\mathbf{w}}}_{f}^{\left(k\right)}$
}and\emph{ $\mathbf{w}_{f}^{\left(k\right)}$} (corresponding to the
DFTs of $\underline{\mathbf{w}}^{\left(k\right)}$ and $\mathbf{w}^{\left(k\right)}$,
respectively) are\emph{ $\underline{\boldsymbol{{\Xi}}}_{k}$ }and\emph{
$\underline{\mathbf{\Theta}}_{k}$}, respectively. \end{prop}
\begin{proof}
The proof of Proposition \ref{pro:Cov_FD_CCI} follows from the fact
that $\underline{\boldsymbol{{\Xi}}}_{k}$ and $\underline{\mathbf{\Theta}}_{k}$
are block circulant and block diagonal matrices.
\end{proof}
Proposition \ref{pro:MultiRound_Cov} indicates that the covariance
of the multi-round CCI plus noise vector can be obtained by separately
computing single-round covariances using Proposition \ref{pro:Cov_singleRound}.
This result greatly impacts the computational complexity of the proposed
algorithm as it will be shown in Subsection \ref{sub:Implementation}.

\subsection{Proposed Scheme\label{sub:ProposedScheme}}

In this subsection, we derive the frequency domain MMSE-based soft
packet combiner that cancels both CAI and ISI jointly over multiple
ARQ rounds in the presence of unknown CCI. 

To combine signals corresponding to ARQ rounds $1,\cdots,k$, we use
conventional soft parallel interference cancellation (PIC) (of both
multi-round CAI and ISI) and unconditional MMSE filtering techniques
\cite{vandendorpe_SigProc_04}. Therefore, at each turbo iteration
of ARQ round $k$, the MMSE-based soft packet combiner produces a
complex scalar decision $z_{t,i}^{\left(k\right)}$ that serves for
computing extrinsic LLR values corresponding to coded and interleaved
bits mapped over symbol $s_{t,i}$. Let $\boldsymbol{{\varphi}}_{t,i}$
denote the vector of \emph{a priori} LLRs of bits corresponding to
symbol $s_{t,i}$, and available at the input of the soft combiner
at a particular turbo iteration. $\sigma_{t,i}^{2}\triangleq\mathbb{{E}}\left[\left|s_{t,i}\right|^{2}\mid\boldsymbol{{\varphi}}_{t,i}\right]-\left\{ \mathbb{{E}}\left[s_{t,i}\mid\boldsymbol{{\varphi}}_{t,i}\right]\right\} ^{2}$
denotes the conditional variance of $s_{t,i}$. By invoking either
the orthogonal projection theorem or Lagrangian methods, and using
(\ref{eq:Block_diag_Mat_SingleRound}) and (\ref{eq:Cov_AllRounds}),
soft MMSE-based packet combining at ARQ round $k$, can be performed
in the frequency domain as, \begin{equation}
\mathbf{z}_{f}^{\left(k\right)}=\boldsymbol{{\Gamma}}^{\left(k\right)}\underline{\mathbf{y}}_{f}^{\left(k\right)}-\boldsymbol{{\Omega}}^{\left(k\right)}\bar{\mathbf{s}}_{f},\label{eq:Forward_Backward_FD}\end{equation}
where $\mathbf{z}_{f}^{\left(k\right)}$ is the DFT of $\mathbf{z}^{\left(k\right)}\triangleq\left[z_{1,0}^{\left(k\right)},\cdots,z_{N_{T},T-1}^{\left(k\right)}\right]^{\top}\in\mathbb{C}^{N_{T}T}$,
i.e., $\mathbf{z}^{\left(k\right)}=\mathbf{U}_{T,N_{T}}^{H}\mathbf{z}_{f}^{\left(k\right)}$,
$\bar{\mathbf{s}}_{f}\in\mathbb{C}^{N_{T}T}$ denotes the DFT of the
soft symbol vector $\bar{\mathbf{s}}\triangleq\mathbb{{E}}\left[\mathbf{s}\mid\boldsymbol{{\varphi}}_{t,i}:\forall\left(t,i\right)\right]$,
and \begin{equation}
\begin{cases}
\boldsymbol{{\Gamma}}^{\left(k\right)} & =\mathrm{diag}\left\{ \underline{\mathbf{\Lambda}}_{0}^{\left(k\right)^{H}}\mathbf{B}_{0}^{\left(k\right)^{-1}},\cdots,\underline{\mathbf{\Lambda}}_{T-1}^{\left(k\right)^{H}}\mathbf{B}_{T-1}^{\left(k\right)^{-1}}\right\} ,\\
\boldsymbol{{\Omega}}^{\left(k\right)} & =\mathbf{C}^{\left(k\right)}-\mathbf{I}_{T}\otimes\mathrm{diag}\left\{ (\tilde{\mathbf{C}}^{\left(k\right)})_{1,1},\cdots,(\tilde{\mathbf{C}}^{\left(k\right)})_{N_{T},N_{T}}\right\} ,\end{cases}\label{eq:Forward_Backward_Expr}\end{equation}
\begin{equation}
\begin{cases}
\mathbf{B}_{i}^{\left(k\right)} & =\underline{\mathbf{\Lambda}}_{i}^{\left(k\right)}\boldsymbol{{\tilde{\Sigma}}}\,\underline{\mathbf{\Lambda}}_{i}^{\left(k\right)^{H}}+\boldsymbol{{\Xi}}_{k},\\
\mathbf{C}_{i}^{\left(k\right)} & =\underline{\mathbf{\Lambda}}_{i}^{\left(k\right)^{H}}\mathbf{B}_{i}^{\left(k\right)^{-1}}\underline{\mathbf{\Lambda}}_{i}^{\left(k\right)},\\
\mathbf{C}^{\left(k\right)} & \triangleq\mathrm{diag}\left\{ \mathbf{C}_{0}^{\left(k\right)},\cdots,\mathbf{C}_{T-1}^{\left(k\right)}\right\} ,\\
\boldsymbol{{\Sigma}}_{i} & \triangleq\mathrm{diag}\left\{ \sigma_{1,i}^{2},\cdots,\sigma_{N_{T},i}^{2}\right\} \mathbb{\in R}^{N_{T}\times N_{T}}.\end{cases}\label{eq:Forward_Backward_Expr_needed_Mats}\end{equation}
Matrices $\tilde{\mathbf{C}}^{\left(k\right)}$ and $\boldsymbol{{\tilde{\Sigma}}}$
denote time averages of $\mathbf{C}_{i}^{\left(k\right)}$ and $\boldsymbol{{\Sigma}}_{i}$,
respectively, as defined in Section \ref{sec:Introduction}. The input
for the soft demapper can be extracted from $\mathbf{z}_{f}^{\left(k\right)}$
as $z_{t,i}^{\left(k\right)}=\mathbf{e}_{t,i}^{\top}\mathbf{U}_{T,N_{T}}^{H}\mathbf{z}_{f}^{\left(k\right)}$
where $\mathbf{e}_{t,i}$ is the $\left(iN_{T}+t\right)$th vector
of the canonical basis. As it can be seen from the forward--backward
filtering structure in (\ref{eq:Forward_Backward_FD}), the frequency
domain MMSE filter explicitly cancels soft CAI and ISI while it only
requires the covariance of unknown CCI plus noise. Note that both
Propositions \ref{pro:MultiRound_Cov} and \ref{pro:Cov_singleRound}
are used to derive (\ref{eq:Forward_Backward_FD}).%
\begin{table*}[t]
\centering{}\caption{\label{tab:SummaryAlgo}Summary of Frequency Domain MMSE-Based Turbo
Packet Combining in the Presence of CCI }
\begin{tabular}{lllll}
\hline 
 & \multicolumn{4}{l}{}\tabularnewline
\textbf{\small 0.} & \multicolumn{4}{l}{\textbf{\small Initialization}}\tabularnewline
 & \multicolumn{4}{l}{{\small Initialize \{$\mathbf{D}_{i}^{\left(0\right)}$\}$\,{}_{i=0}^{T-1}$
and $\mathbf{\underline{\tilde{y}}}_{f}^{\left(0\right)}$ with $\mathbf{0}_{N_{T}\times N_{T}}$
and $\mathbf{0}_{N_{T}\times1}$, respectively.}}\tabularnewline
\textbf{\small 1.} & \multicolumn{4}{l}{\textbf{\small Packet Combining at ARQ round $\boldsymbol{{k}}$}}\tabularnewline
 & \textbf{\small 1.1.} & \multicolumn{3}{l}{{\small Compute the DFT of CSI and received signals at round $k$,
i.e., $\mathbf{\Lambda}_{0}^{\left(k\right)},\cdots,\mathbf{\Lambda}_{T-1}^{\left(k\right)}$
and $\mathbf{y}_{f}^{\left(k\right)}$, respectively.}}\tabularnewline
 & \textbf{\small 1.2.} & \multicolumn{3}{l}{{\small For each turbo iteration}}\tabularnewline
 &  & \textbf{\small 1.2.1.} & \multicolumn{2}{l}{{\small Compute the soft symbol vector $\bar{\mathbf{s}}$ and variances
$\sigma_{t,i}^{2}$, then deduce the DFT $\bar{\mathbf{s}}_{f}$ and
$\boldsymbol{\tilde{{\Sigma}}}$.}}\tabularnewline
 &  & \textbf{\small 1.2.2.} & \multicolumn{2}{l}{{\small Estimate the CCI plus noise covariance $\mathbf{\Theta}_{k}$
at round $k$ using (\ref{eq:cov_estimate_round_k}), then compute
$\mathbf{\Theta}_{k}^{-1}$.}}\tabularnewline
 &  & \textbf{\small 1.2.3.} & \multicolumn{2}{l}{{\small Compute \{$\mathbf{\Lambda}_{i}^{\left(k\right)^{H}}\mathbf{\Theta}_{k}^{-1}\mathbf{\Lambda}_{i}^{\left(k\right)}$\}$\,{}_{i=0}^{T-1}$
and $\mathbf{\Lambda}^{\left(k\right)^{H}}(\mathbf{I}_{T}\otimes\mathbf{\Theta}_{k}^{-1})\mathbf{y}_{f}^{\left(k\right)}$. }}\tabularnewline
 &  & \textbf{\small 1.2.4.} & \multicolumn{2}{l}{{\small Deduce \{$\mathbf{D}_{i}^{\left(k\right)}$\}$\,{}_{i=0}^{T-1}$
and $\mathbf{\underline{\tilde{y}}}_{f}^{\left(k\right)}$ using recursions
(\ref{eq:recursion_Di(k)}) and (\ref{eq:recursion2}), respectively.}}\tabularnewline
 &  & \textbf{\small 1.2.5.} & \multicolumn{2}{l}{{\small Compute the inverses of \{$\boldsymbol{{\tilde{\Sigma}}}+\mathbf{D}_{i}^{\left(k\right)}$\}$\,{}_{i=0}^{T-1}$,
then construct \{$\mathbf{C}_{i}^{\left(k\right)}$\}$\,{}_{i=0}^{T-1}$
and $\mathbf{F}^{\left(k\right)}$ using (\ref{eq:Matrix_Ci(k)_2ndExpr}) }}\tabularnewline
 &  &  & \multicolumn{2}{l}{{\small and (\ref{eq:EVD}), respectively. }}\tabularnewline
 &  & \textbf{\small 1.2.6.} & \multicolumn{2}{l}{{\small Compute forward and backward parts $\boldsymbol{{\Gamma}}^{\left(k\right)}\underline{\mathbf{y}}_{f}^{\left(k\right)}$
and $\boldsymbol{{\Omega}}^{\left(k\right)}\bar{\mathbf{s}}_{f}$
using (\ref{eq:Forward_Filtering}) and (\ref{eq:Forward_Backward_Expr}),
respectively. }}\tabularnewline
 &  & \textbf{\small 1.2.7.} & \multicolumn{2}{l}{{\small Calculate extrinsic LLRs.}}\tabularnewline
 &  & \textbf{\small 1.2.8.} & \multicolumn{2}{l}{{\small Perform SISO decoding.}}\tabularnewline
 & \textbf{\small 1.3.} & \textbf{\small 1.2.9.} & \multicolumn{2}{l}{{\small End }\textbf{\small 1.2.}}\tabularnewline
 & \textbf{\small 1.4.} & \multicolumn{3}{l}{{\small If {}``frame error'' then store \{$\mathbf{D}_{i}^{\left(k\right)}$\}$\,{}_{i=0}^{T-1}$
and $\mathbf{\underline{\tilde{y}}}_{f}^{\left(k\right)}$ computed
at the last turbo iteration, and send {}``NACK'',}}\tabularnewline
 &  & \multicolumn{3}{l}{{\small Otherwise, send {}``ACK''.}}\tabularnewline
 &  &  &  & \tabularnewline
\hline
\end{tabular}
\end{table*}

To obtain estimates of unknown CCI plus noise covariance matrices
$\mathbf{\Theta}_{1},\cdots,\mathbf{\Theta}_{k}$, required by (\ref{eq:Forward_Backward_Expr_needed_Mats}),
let us consider the single-round frequency domain communication model
(\ref{eq:FD_CommModel_oneARQround}). Proposition \ref{pro:Cov_FD_CCI}
indicates that the covariance of $\mathbf{w}_{f}^{\left(k\right)}$
is $\underline{\mathbf{\Theta}}_{k}=\mathbf{I}_{T}\otimes\mathbf{\Theta}_{k}.$
Therefore, with respect to the block diagonal structure of (\ref{eq:FD_CommModel_oneARQround}),
unknown CCI plus noise covariance $\mathbf{\Theta}_{k}$ can directly
be estimated in the frequency domain at each turbo iteration, with
the aid of \emph{a priori} LLRs, according to the following average,
\begin{equation}
\mathbf{\Theta}_{k}=\frac{1}{T}\sum_{i=0}^{T-1}\left\{ \mathbf{y}_{f_{i}}^{\left(k\right)}-\mathbf{\Lambda}_{i}^{\left(k\right)}\bar{\mathbf{s}}_{f_{i}}\right\} \left\{ \mathbf{y}_{f_{i}}^{\left(k\right)}-\mathbf{\Lambda}_{i}^{\left(k\right)}\bar{\mathbf{s}}_{f_{i}}\right\} ^{H}.\label{eq:cov_estimate_round_k}\end{equation}
$\mathbf{y}_{f_{i}}^{\left(k\right)}$ and $\bar{\mathbf{s}}_{f_{i}}$
denote the DFTs of $\mathbf{y}^{\left(k\right)}$ and $\mathbf{s}$
at frequency bin $i$, respectively, i.e., $\mathbf{y}_{f_{i}}^{\left(k\right)}=\mathbf{E}_{i,N_{R}}\mathbf{y}_{f}^{\left(k\right)}$
and $\bar{\mathbf{s}}_{f_{i}}=\mathbf{E}_{i,N_{T}}\bar{\mathbf{s}}_{f}$.
Covariance matrices $\mathbf{\Theta}_{1},\cdots,\mathbf{\Theta}_{k-1}$
are similarly estimated at ARQ rounds $1,\cdots,k-1$, respectively,
and correspond to estimates obtained at the last turbo iteration.
In other words, when the decoding outcome is erroneous, a NACK message
is fed back to the transmitter, and the unknown CCI plus noise covariance
estimate obtained at the last iteration is saved in the receiver to
help perform packet combining at the next ARQ round.

\subsection{Implementation Aspects\label{sub:Implementation}}

We first provide an efficient implementation of the proposed scheme
since turbo combining requires at each turbo iteration the computation
of matrix inverses $\mathbf{B}_{0}^{\left(k\right)^{-1}},\cdots,\mathbf{B}_{T-1}^{\left(k\right)^{-1}}\in\mathbb{C}^{kN_{R}\times kN_{R}}$
given by (\ref{eq:Forward_Backward_Expr_needed_Mats}). Second, we
analyze the computational complexity and memory requirements of the
proposed implementation algorithm.

\subsubsection{An Efficient Implementation Algorithm}

The special structure of the frequency domain ARQ channel matrix (\ref{eq:CSI_DFT_allRounds})
together with the matrix inversion lemma \cite{Haykin} allow us to
express the inverse of $\mathbf{B}_{i}^{\left(k\right)}$ as, \begin{equation}
\mathbf{B}_{i}^{\left(k\right)^{-1}}=\boldsymbol{{\Xi}}_{k}^{-1}-\boldsymbol{{\Xi}}_{k}^{-1}\underline{\mathbf{\Lambda}}_{i}^{\left(k\right)}\left(\boldsymbol{{\tilde{\Sigma}}}+\mathbf{D}_{i}^{\left(k\right)}\right)^{-1}\underline{\mathbf{\Lambda}}_{i}^{\left(k\right)^{H}}\boldsymbol{{\Xi}}_{k}^{-1},\label{eq:Inversematrix_Bi(k)}\end{equation}
where $\mathbf{D}_{i}^{\left(k\right)}$ is obtained according to
the following recursion, \begin{equation}
\begin{cases}
\mathbf{D}_{i}^{\left(k\right)} & =\mathbf{D}_{i}^{\left(k-1\right)}+\mathbf{\Lambda}_{i}^{\left(k\right)^{H}}\mathbf{\Theta}_{k}^{-1}\mathbf{\Lambda}_{i}^{\left(k\right)},\\
\mathbf{D}_{i}^{\left(0\right)} & =\mathbf{0}_{N_{T}\times N_{T}}.\end{cases}\label{eq:recursion_Di(k)}\end{equation}
Therefore, matrices $\mathbf{C}_{0}^{\left(k\right)},\cdots,\mathbf{C}_{T-1}^{\left(k\right)}$
are simply computed as, \begin{equation}
\mathbf{C}_{i}^{\left(k\right)}=\mathbf{D}_{i}^{\left(k\right)}-\mathbf{D}_{i}^{\left(k\right)}\left(\boldsymbol{{\tilde{\Sigma}}}+\mathbf{D}_{i}^{\left(k\right)}\right)^{-1}\mathbf{D}_{i}^{\left(k\right)},\label{eq:Matrix_Ci(k)_2ndExpr}\end{equation}
while the forward filtering part of (\ref{eq:Forward_Backward_FD})
is calculated at each ARQ round $k$ as, \begin{equation}
\boldsymbol{{\Gamma}}^{\left(k\right)}\underline{\mathbf{y}}_{f}^{\left(k\right)}=\mathbf{F}^{\left(k\right)}\mathbf{\underline{\tilde{y}}}_{f}^{\left(k\right)},\label{eq:Forward_Filtering}\end{equation}
where\begin{multline}
\mathbf{F}^{\left(k\right)}=\mathrm{diag}\left\{ \mathbf{I}_{N_{T}}-\mathbf{D}_{0}^{\left(k\right)}\left(\boldsymbol{{\tilde{\Sigma}}}+\mathbf{D}_{0}^{\left(k\right)}\right)^{-1},\cdots,\right.\,\,\,\,\,\,\,\,\,\,\,\,\,\,\,\,\\
\left.\mathbf{I}_{N_{T}}-\mathbf{D}_{T-1}^{\left(k\right)}\left(\boldsymbol{{\tilde{\Sigma}}}+\mathbf{D}_{T-1}^{\left(k\right)}\right)^{-1}\right\} ,\label{eq:Matrix_F(k)}\end{multline}
 and $\mathbf{\underline{\tilde{y}}}_{f}^{\left(k\right)}$ is given
by the following recursion, \begin{equation}
\begin{cases}
\mathbf{\underline{\tilde{y}}}_{f}^{\left(k\right)} & =\mathbf{\underline{\tilde{y}}}_{f}^{\left(k-1\right)}+\mathbf{\Lambda}^{\left(k\right)^{H}}(\mathbf{I}_{T}\otimes\mathbf{\Theta}_{k}^{-1})\mathbf{y}_{f}^{\left(k\right)},\\
\mathbf{\underline{\tilde{y}}}_{f}^{\left(0\right)} & =\mathbf{0}_{N_{T}\times1}.\end{cases}\label{eq:recursion2}\end{equation}
 The proposed turbo packet combining algorithm is summarized in Table
\ref{tab:SummaryAlgo}. Note that, during the first iteration of round
$k$, the anti-causal parts in recursions (\ref{eq:recursion_Di(k)})
and (\ref{eq:recursion2}), i.e., $\mathbf{D}_{i}^{\left(k-1\right)}$
and $\mathbf{\underline{\tilde{y}}}_{f}^{\left(k-1\right)}$, respectively,
correspond to the output of these recursions at the last iteration
of previous round $k-1$.

\subsubsection{Computational Complexity and Memory Requirements}

The proposed recursive implementation algorithm avoids storing received
signals and CFRs corresponding to multiple ARQ rounds. It also prevents
the computation of $kN_{R}\times kN_{R}$ matrix inverses. This dramatically
reduces the implementation cost since the complexity order of directly
computing $\mathbf{B}_{i}^{\left(k\right)^{-1}}$ is cubic against
$kN_{R}$, and is greatly increased from round to round. In the following,
we analyze both the complexity and memory requirements of the proposed
scheme, and compare them with those of the LLR-level combining technique
\footnote{In this paper, LLR-level combining refers to the iterative (turbo)
packet combining and SISO decoding receiver, where transmissions corresponding
to $k$ ARQ rounds are separately turbo equalized using $k$ frequency
domain MMSE soft equalizers. To perform packet combining at each iteration
of ARQ round $k$, extrinsic LLR values generated by the soft MMSE
equalizer at round $k$ and those obtained at the last iteration of
previous rounds $1,\cdots,k-1$ are added together, then SISO decoding
is performed. %
}. %
\begin{table*}[t]
\begin{centering}
\caption{\label{tab:SummaryImplementation}Summary of Memory and Arithmetic
Additions Required by the Proposed and LLR-Level Combining Schemes,
and Relative Cost Evaluation}

\par\end{centering}

\centering{}\begin{tabular}{lccc}
\hline 
\textbf{\footnotesize Combining scheme} & \textbf{\footnotesize Memory} & \textbf{\footnotesize Arithmetic Additions} & \begin{tabular}{c}
\textbf{\footnotesize Relative Costs }\tabularnewline
\textbf{\footnotesize }\begin{tabular}{lll}
\hline 
{\footnotesize QPSK} & {\footnotesize 8-PSK~~} & {\footnotesize 16-QAM}\tabularnewline
\end{tabular}\tabularnewline
\end{tabular}\tabularnewline
\hline
\hline 
\begin{tabular}{l}
{\footnotesize LLR-Level}\tabularnewline
{\footnotesize Proposed}\tabularnewline
\end{tabular} & \begin{tabular}{c}
\textbf{\footnotesize $TN_{T}\log_{2}\left|\mathcal{S}\right|$}\tabularnewline
\textbf{\footnotesize $2TN_{T}\left(N_{T}+1\right)$}\tabularnewline
\end{tabular} & \begin{tabular}{c}
\textbf{\footnotesize $TN_{T}N_{it}\left(K-1\right)\log_{2}\left|\mathcal{S}\right|$}\tabularnewline
\textbf{\footnotesize $2TN_{T}N_{it}\left(K-1\right)\left(N_{T}+1\right)$}\tabularnewline
\end{tabular} & \begin{tabular}{ccc}
{\footnotesize ~~$N_{T}$} & {\footnotesize $\frac{2}{3}N_{T}-\frac{1}{3}$} & {\footnotesize ~$N_{T}-\frac{1}{2}$}\tabularnewline
\end{tabular}\tabularnewline
\hline
\end{tabular}
\end{table*}
\begin{figure*}[!t]
\normalsize 
\vspace*{6pt}
\hrulefill
\vspace*{4pt}
\setcounter{mytempeqncnt}{\value{equation}} 
\setcounter{equation}{41}
\begin{equation}
\mathbf{R}_{N_{R}}=\left[\begin{array}{ccc} 1 &  & \delta_{\mathrm{Rx}}\\  & \ddots\\ \delta_{\mathrm{Rx}} &  & 1\end{array}\right]_{N_{R}\times N_{R}},\,\,\,\,\,\mathbf{R}_{N'_{T}}=\left[\begin{array}{ccc} 1 &  & \delta_{\mathrm{Tx}}\\  & \ddots\\ \delta_{\mathrm{Tx}} &  & 1\end{array}\right]_{N'_{T}\times N'_{T}},\label{eq:CorrMat_Tx_Rx_CorrCh}\end{equation}
\setcounter{equation}{38} 
\hrulefill 
\vspace*{4pt}
\end{figure*}

First, note that in the case of LLR-level packet combining, frequency
domain MMSE equalization is separately performed for each ARQ round.
Therefore, $T$ inversions of $N_{R}\times N_{R}$ matrices are required
to compute the forward and backward filters. Since in general it is
required to have more receive than transmit antennas, especially when
CCI is present in the system, an implementation similar to that introduced
in the previous subsection is beneficial because only $T$ inversions
of $N_{T}\times N_{T}$ matrices will be required. In this case, the
two variables in recursions (\ref{eq:recursion_Di(k)}) and (\ref{eq:recursion2})
are computed at ARQ round $k$ as, $\mathbf{D}_{i}^{'\left(k\right)}=\mathbf{\Lambda}_{i}^{\left(k\right)^{H}}\mathbf{\Theta}_{k}^{-1}\mathbf{\Lambda}_{i}^{\left(k\right)}\in\mathbb{C}^{N_{T}\times N_{T}},$
and $\mathbf{\underline{\tilde{y}}}_{f}^{'\left(k\right)}=\mathbf{\Lambda}^{\left(k\right)^{H}}(\mathbf{I}_{T}\otimes\mathbf{\Theta}_{k}^{-1})\mathbf{y}_{f}^{\left(k\right)}\in\mathbb{C}^{N_{T}},$
while all the other steps in Table \ref{tab:SummaryAlgo} remain the
same (including the CCI plus noise covariance estimation procedure
in step 1.2.2.). Therefore, by letting $N_{it}$ denote the number
of turbo iterations at each ARQ round, both combining algorithms have
similar computational complexities since the proposed scheme and the
LLR-level scheme require at most $C_{new\, scheme}^{+}=2TN_{T}N_{it}\left(K-1\right)\left(N_{T}+1\right)$
and $C_{LLR-level}^{+}=TN_{T}N_{it}\left(K-1\right)\log_{2}\left|\mathcal{S}\right|$
arithmetic additions to perform (\ref{eq:recursion_Di(k)}) and (\ref{eq:recursion2}),
and to combine LLRs corresponding to multiple rounds, respectively. 

LLR-level packet combining performs the combination of extrinsic LLR
values generated by frequency domain soft equalizers at multiple ARQ
rounds. Therefore, a storage capacity of $TN_{T}\log_{2}\left|\mathcal{S}\right|$
real values is required to store accumulated LLR values corresponding
to all ARQ rounds. The proposed scheme combines multiple transmissions
at the signal level using signals and CFRs corresponding to all ARQ
rounds, without being required to be explicitly stored in the receiver.
This is performed with the aid of the two variables $\mathbf{D}_{i}^{\left(k\right)}$
and $\mathbf{\underline{\tilde{y}}}_{f}^{\left(k\right)}$ in recursions
(\ref{eq:recursion_Di(k)}) and (\ref{eq:recursion2}), respectively.
This translates into a memory size of $2TN_{T}\left(N_{T}+1\right)$
real values. Therefore, the computational complexity and storage requirements
are less sensitive to the ARQ delay. The technique requires only a
few more additions and a bit more memory compared to LLR-level combining.
Table \ref{tab:SummaryImplementation} summarizes implementation requirements
and reports the relative costs %
\footnote{Relative costs refer to the relative number of arithmetic additions
$\triangle C$ and memory $\triangle M$ required by the proposed
scheme compared to LLR-level combining. With respect to storage requirements
and number of arithmetic additions in Table \ref{tab:SummaryImplementation},
we have $\triangle C=\triangle M=2\frac{N_{T}+1}{\log_{2}\left|\mathcal{S}\right|}$.%
} for some modulation schemes.

\section{Performance Evaluation\label{sec:Performance-Evaluation}}

\subsection{Asymptotic Performance Analysis\label{sub:Asymptotic-Analysis}}

In the following, we provide a frame-basis analysis where we derive
system conditions under which perfect CCI cancellation holds. We suppose
that the interferer CSI is perfectly known, and investigate the influence
of its channel properties on the interference cancellation capability
of the proposed packet combining scheme in the high SNR regime. 
\begin{thm}
\label{thm:CCI_supress_prop}We consider a CCI-limited MIMO ARQ system
with $N_{T}$ transmit and $N_{R}$ receive antennas, and ARQ delay
$K$. Let $\boldsymbol{{\Theta}}_{k}^{\mathrm{CCI}}$ denote the CCI
covariance at ARQ round $k=1,\cdots,K$, i.e., the covariance of the
global noise at the receiver is $\boldsymbol{{\Theta}}_{k}=\boldsymbol{{\Theta}}_{k}^{\mathrm{CCI}}+\sigma^{2}\mathbf{I}_{N_{R}}$,
and $\rho_{k}$ be the rank of $\boldsymbol{{\Theta}}_{k}^{\mathrm{CCI}}$.
We assume perfect LLR feedback from the SISO decoder. The frequency
domain soft MMSE packet combiner provides perfect CCI suppression
for asymptotically high SNR if \begin{equation}
\sum_{u=1}^{k}\rho_{u}<kN_{R}-N_{T}.\label{eq:Rank-Condition}\end{equation}
\end{thm}
\begin{proof}
See the Appendix.
\end{proof}
We now proceed to derive an upper bound on $\rho_{k}$, where we incorporate
the rank of the CCI fading channel. Under the assumption that CCI
symbols satisfy (\ref{eq:Symb_independency}), i.e., infinitely deep
interleaving, we get \begin{equation}
\boldsymbol{{\Theta}}_{k}^{\mathrm{CCI}}=\sum_{l'=0}^{L'-1}\mathbf{H}_{l'}^{\mathrm{CCI}^{\left(k\right)}}\mathbf{H}_{l'}^{\mathrm{CCI}^{\left(k\right)^{H}}}.\label{eq:CCI_Cov_Mat_Expr2}\end{equation}
Let us write each CCI channel matrix as \begin{equation}
\mathbf{H}_{l'}^{\mathrm{CCI}^{\left(k\right)}}=\mathbf{R}_{N_{R}}^{1/2}\mathbf{A}_{l'}^{\mathrm{CCI}^{\left(k\right)}}\mathbf{R}_{N'_{T}}^{1/2}\,\,\,\,\,\,\,\forall l',\label{eq:CCI_Channel_Gesbert}\end{equation}
where $\mathbf{A}_{l'}^{\mathrm{CCI}^{\left(k\right)}}\in\mathbb{C}^{N_{R}\times N'_{T}}$
characterizes the scattering environment between the CCI transmitter
and receiver \cite{Gesbert_Corr_Tcom02}, and $\mathbf{R}_{N_{R}}$
and $\mathbf{R}_{N'_{T}}$ are the correlation matrices controlling
the receive and transmit antenna arrays, and are in general given
by (\ref{eq:CorrMat_Tx_Rx_CorrCh}), \setcounter{equation}{42}where
$0\leq\delta_{\mathrm{Rx}},\,\delta_{\mathrm{Tx}}<1$ \cite{Sellathurai_Corr_JSAC03}.
Note that (\ref{eq:CCI_Channel_Gesbert}) corresponds to a general
model of correlated fading MIMO channels, where the scattering radii
at transmitter and receiver sides is taken into account, and $\mathbf{A}_{l'}^{\mathrm{CCI}^{\left(k\right)}}$
is not necessarily a full rank matrix, i.e., $\mathrm{rank}\left\{ \mathbf{A}_{l'}^{\mathrm{CCI}^{\left(k\right)}}\right\} \leq\min\left(N'_{T},N_{R}\right)$
\cite{Gesbert_Corr_Tcom02}. Noting that $\mathbf{R}_{N_{R}}$ and
$\mathbf{R}_{N'_{T}}$ are full rank matrices, and with respect to
the fact that CCI tap channel matrices are independent, and using
(\ref{eq:CCI_Cov_Mat_Expr2}) and (\ref{eq:CCI_Channel_Gesbert}),
we get

\begin{align}
\rho_{k} & \leq\min\left\{ N_{R},\sum_{l'=0}^{L'-1}\mathrm{rank}\left\{ \mathbf{H}_{l'}^{\mathrm{CCI}^{\left(k\right)}}\mathbf{H}_{l'}^{\mathrm{CCI}^{\left(k\right)^{H}}}\right\} \right\} \nonumber \\
 & =\min\left\{ N_{R},\sum_{l'=0}^{L'-1}\mathrm{rank}\left\{ \mathbf{A}_{l'}^{\mathrm{CCI}^{\left(k\right)}}\mathbf{R}_{N'_{T}}\mathbf{A}_{l'}^{\mathrm{CCI}^{\left(k\right)^{H}}}\right\} \right\} \nonumber \\
 & \leq\min\left\{ N_{R},\sum_{l'=0}^{L'-1}\mathrm{rank}\left\{ \mathbf{A}_{l'}^{\mathrm{CCI}^{\left(k\right)}}\right\} \right\} .\label{eq:CCI_cov_rank_UpperBound}\end{align}

A closer look at Theorem \ref{thm:CCI_supress_prop} and upper bound
(\ref{eq:CCI_cov_rank_UpperBound}) provides interesting system interpretations.
\begin{itemize}
\item \textbf{Impact of CCI Fading Channel:} First, note that the CCI cancellation
capability of the frequency domain MMSE packet combiner is related
to the CCI channel rank. When the interferer has a rank-deficient
channel matrix at a certain ARQ round, interference can completely
be removed (at subsequent rounds) if the sum-rank condition in Theorem
\ref{thm:CCI_supress_prop} is satisfied. In practice, the channel
rank can dramatically drop in the case of the so-called pinhole channel,
where the transmitter and receiver are largely separated and are surrounded
by multiple scatterers \cite{Gesbert_Corr_Tcom02}. In this scenario,
the channel can even prevent multipath from building up since the
thin air pipe connecting transmitter and receiver scatterers is very
long. For instance, in a system with $N_{R}=3$ receive and $N_{T}=2$
transmit antennas, and an unknown interferer who is experiencing one
path ($L'=1$) channel realizations with rank equal to two, CCI can
be removed at the second ARQ round because the sum-rank condition
(\ref{eq:Rank-Condition}) holds for $k\geq2$.
\item \textbf{Impact of the Number of Transmit Antennas and ARQ Delay: }Condition
(\ref{eq:Rank-Condition}) suggests how, for a given CCI channel profile,
the number of transmit antennas $N_{T}$ and ARQ rounds $K$ are chosen
to achieve perfect CCI cancellation. For instance, if transmission
is corrupted by CCI with quasi-static channel rank %
\footnote{In this case, CCI with quasi-static channel rank refers to an interferer
whose channel rank is constant over multiple ARQ rounds.%
}, and if the ARQ delay allowed by the upper layer is $K$, then only
$N_{T}<K\left(N_{R}-\rho_{0}\right)$ transmit antennas can be allocated
to the user of interest to achieve interference suppression at the
latest at ARQ round $K$, where $\rho_{0}$ is the rank of $\boldsymbol{{\Theta}}_{k}^{\mathrm{CCI}}$,
i.e., $\rho_{k}=\rho_{0}\,\,\forall k$. Increasing the ARQ delay
will relax the condition on the number of transmit antennas and therefore
allow for an increase in the diversity and/or multiplexing gains depending
on the diversity-multiplexing-delay trade-off operating point \cite{DMD_Gamal_IT_2006}.
Note that when $N'_{T}\ll N_{T}$, the CCI channel rank dramatically
drops, and therefore CCI suppression is achieved even when a short
ARQ delay $K$ is required. 
\item \textbf{Interaction with the Scheduling Mechanism: }In the case of
opportunistic communications, interference with co-channel users who
have high channel ranks can be prevented. For instance, when a retransmission
is required on the reverse link, the base station (BS) can choose
the timing of the next ARQ round in such a way that transmission simultaneously
occurs with that of a user with low channel rank. This is feasible
since the BS has complete knowledge about user CSIs in the reverse
link. The same scheduling mechanism can be used in the forward link
if all users provide the BS with feedback information about their
channel ranks. When the system suffers from CCI caused by neighboring
cells, the sum-rank condition (\ref{eq:Rank-Condition}) can be achieved
by simply increasing the number of ARQ rounds because the CCI channel
rank tends to be constant over time. 
\end{itemize}

\subsection{Numerical Results\label{sub:BLER_throughput-Performance}}

In this subsection, we provide block error rate (BLER) performance
results for the proposed combining technique. Our focus is to demonstrate
the superior performance of the introduced scheme compared to LLR-level
combining. We also evaluate BLER performance for scenarios where the
interferer has rank deficient channel matrices to corroborate the
theoretical analysis in Subsection \ref{sub:Asymptotic-Analysis}.%
\begin{figure}[t]
\noindent \begin{centering}
\includegraphics[scale=0.49]{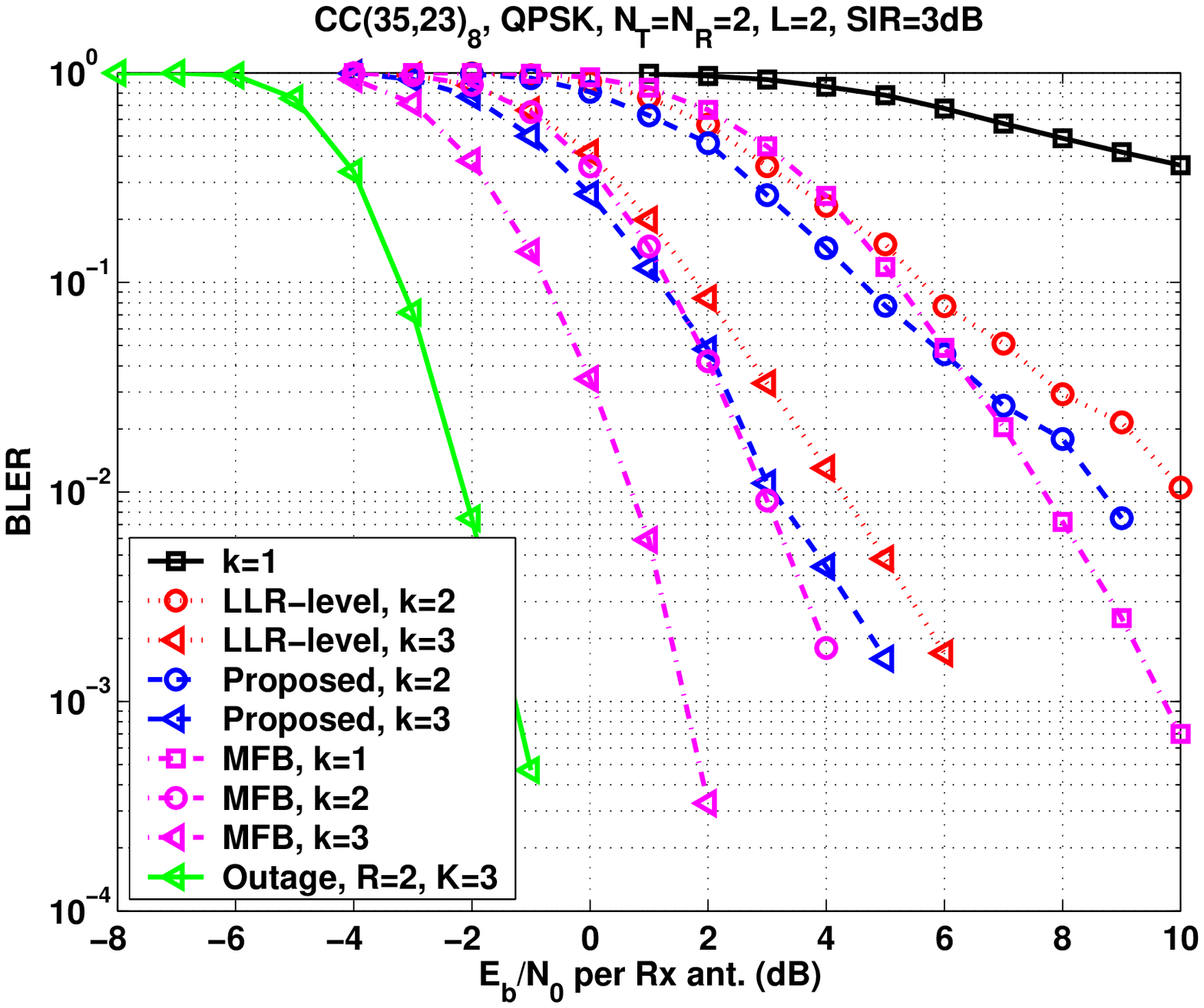}
\par\end{centering}

\caption{\label{fig:BLER_2x2_SIR3dB} BLER performance for CC $\left(35,23\right)_{8}$,
QPSK, $N_{T}=N_{R}=2$, $L=L'=2$ equal energy paths, and $\mathrm{SIR}=3$dB.}

\end{figure}
\begin{figure}[t]
\noindent \begin{centering}
\includegraphics[scale=0.49]{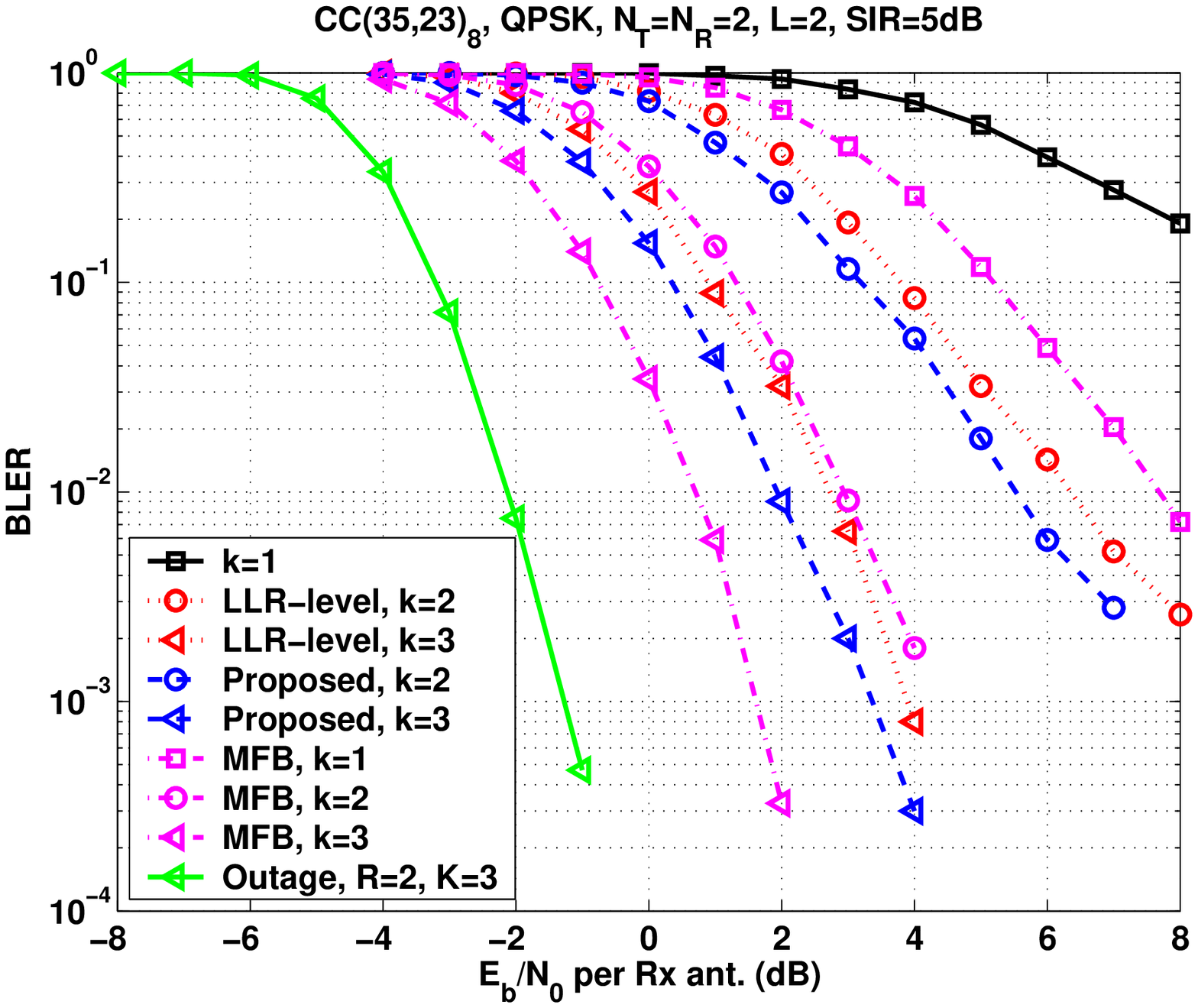}
\par\end{centering}

\caption{\label{fig:BLER_2x2_SIR5dB} BLER performance for CC $\left(35,23\right)_{8}$,
QPSK, $N_{T}=N_{R}=2$, $L=L'=2$ equal energy paths, and $\mathrm{SIR}=5$dB.}

\end{figure}

In all simulations, we consider a BICM scheme where the encoder is
a $\frac{1}{2}$-rate convolutional code with polynomial generators
$\left(35,23\right)_{8}$, and the modulation scheme is quadrature
phase shift keying (QPSK). The length of the code bit frame is $1032$
bits including tails. The ARQ delay is $K=3$, and the $E_{b}/N_{0}$
ratio appearing in all figures is the SNR per useful bit per receive
antenna. We consider a $L=2$ path MIMO-ISI channel profile where
$\sigma_{0}^{2}=\sigma_{1}^{2}=\frac{1}{2}$. In practical wireless
systems, the wireless channel may have more than two paths due to
severe frequency selective fading. In this paper, we restrict ourselves
to $L=2$ for the sake of simulation simplicity. Performance in the
case of severe frequency selective fading channels can be found in
\cite{Chafnaji_PIMRC08}. We use both the matched filter bound (MFB)
per ARQ round and the outage probability \cite{Ait-Idir_MIMOISIARQ_TCOM08}
of the CCI-free MIMO-ISI channel as absolute performance bounds to
evaluate the CCI cancellation capability and diversity order achieved
by the proposed combining scheme. The number of turbo iterations is
set to five and the Max-Log-MAP version of the maximum \emph{a posteriori}
(MAP) algorithm is used for SISO decoding.

We first investigate performance for scenarios where the user of interest
and the interferer have the same number of transmit antennas ($N_{T}=N'_{T}$)
and identical channel profiles, i.e., $L=L'$, equal power taps, and
CCI fading channel coefficients are i.i.d. In Fig. \ref{fig:BLER_2x2_SIR3dB},
we compare the BLER performance of the proposed scheme with that of
LLR-level combining for a ST--BICM code with rate $R=2$, i.e., $N_{T}=2$.
The number of receive antennas is $N_{R}=2$, and $\mathrm{SIR}=3\mathrm{dB}$.
We observe that the proposed scheme significantly outperforms LLR-level
combining. The performance gap at ARQ round $k=3$ is about $1\mathrm{dB}$
for $\mathrm{BLER}\leq10^{-2}$. Note that both combining schemes
fail to perfectly cancel CCI since performance curves tend to saturate
for high $E_{b}/N_{0}$ values. Fig. \ref{fig:BLER_2x2_SIR5dB} reports
performance of both techniques when $\mathrm{SIR}$ is increased to
$\mathrm{SIR}=5\mathrm{dB}$. In this case, the performance gap between
the two schemes is reduced. The CCI cancellation capability is also
improved as can be seen from the steeper slopes of BLER curves. In
Fig. \ref{fig:BLER_4x2_SIR5dB}, we evaluate the performance for a
high rate ST--BICM code where $R=4$, i.e., $N_{T}=4$. Only $N_{R}=2$
receive antennas are considered, and $\mathrm{SIR}=5\mathrm{dB}$.
The proposed scheme dramatically outperforms LLR-level combining,
i.e., the performance gap at ARQ round $k=3$ is about $4\mathrm{dB}$
at $7*10^{-3}$ BLER. The proposed scheme also offers higher cancellation
capability and diversity order than LLR-level combining.%
\begin{figure}[t]
\noindent \begin{centering}
\includegraphics[scale=0.49]{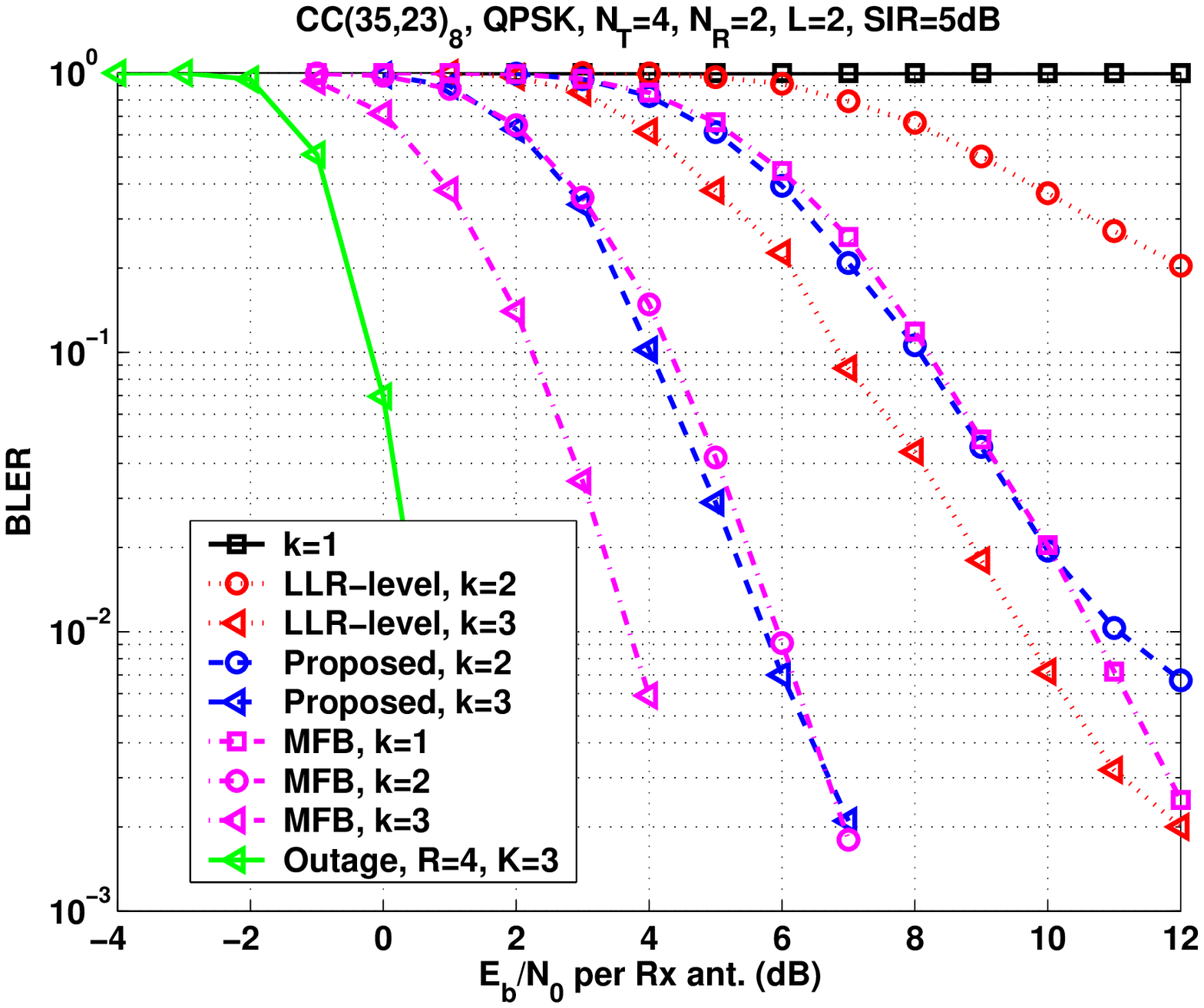}
\par\end{centering}

\caption{\label{fig:BLER_4x2_SIR5dB} BLER performance for CC $\left(35,23\right)_{8}$,
QPSK, $N_{T}=4$, $N_{R}=2$, $L=L'=2$ equal energy paths, and $\mathrm{SIR}=5$dB.}

\end{figure}
\begin{figure}[t]
\noindent \begin{centering}
\includegraphics[scale=0.49]{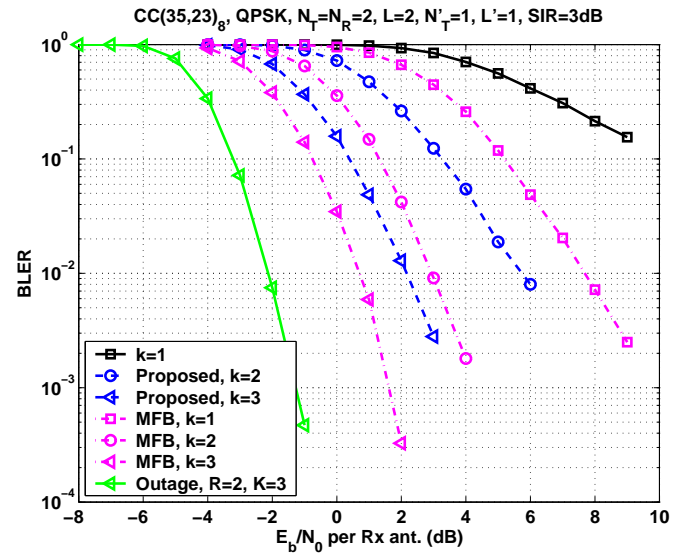}
\par\end{centering}

\caption{\label{fig:BLER_2x2_NpT1_Lp1SIR3dB} BLER performance for CC $\left(35,23\right)_{8}$,
QPSK, $N_{T}=N_{R}=2$, $L=2$ equal energy paths, $N'_{T}=1$, $L'=1$,
and $\rho_{k}=1,\, k=1,\cdots,K$, $\mathrm{SIR}=3$dB.}

\end{figure}

We now turn to scenarios where the interferer has a rank-deficient
uncorrelated MIMO channel, i.e., $\mathrm{rank}\left\{ \mathbf{A}_{l'}^{\mathrm{CCI}^{\left(k\right)}}\right\} <\min\left(N'_{T},N{}_{R}\right)\,\,\forall l'$,
$\delta_{\mathrm{Tx}}=\delta_{\mathrm{Rx}}=0$, and assume that the
rank is constant over all ARQ rounds. In Fig. \ref{fig:BLER_2x2_NpT1_Lp1SIR3dB},
we report the BLER performance of the proposed scheme for a CCI-limited
MIMO system with settings similar to Fig. \ref{fig:BLER_2x2_SIR3dB},
i.e., $N_{T}=N_{R}=2$, and $\mathrm{SIR}=3\mathrm{dB}$. The interferer
experiences flat fading, i.e., $L'=1$, and only has $N'_{T}=1$ transmit
antenna. Therefore, with respect to (\ref{eq:CCI_cov_rank_UpperBound}),
$\rho_{k}=1\,\,\forall k$. Note that in this interference scenario,
the perfect CCI cancellation condition (\ref{eq:Rank-Condition})
holds for $k\geq2$. We observe that both the CCI cancellation capability
and the diversity order of the proposed scheme are improved. The performance
gain with respect to the case of $N'_{T}=2$ and $L'=2$ is about
$1.5\mathrm{dB}$ at $3*10^{-3}$ BLER and round $k=3$, and the slope
of the BLER curve at round $k=3$ is similar to that of the MFB curve.
Fig. \ref{fig:BLER_4x4_SIR1dB} compares the performance of the proposed
scheme for two scenarios with heavy CCI, i.e., $\mathrm{SIR}=1\mathrm{dB}$.
The ST--BICM code has rate $R=4$, i.e., $N_{T}=4$, and the number
of receive antennas is set to $N_{R}=4$. In the first scenario (Scenario
1), the interferer has $N'_{T}=4$ transmit antennas, $L'=2$ equal
power taps, and i.i.d. fading coefficients, while in the second scenario
(Scenario 2), $N'_{T}=2$, $L'=1$, and the CCI channel rank is equal
to two. Therefore, $\rho_{k}=2\,\,\forall k$, and condition (\ref{eq:Rank-Condition})
holds for $k\geq2$. It is clear that in the second scenario, better
CCI cancellation capability is achieved for $k\geq2$. For instance,
the performance gap for $k=3$ is more than $2\mathrm{dB}$ at $2*10^{-2}$
BLER. Also, the diversity order of the CCI-free MIMO-ISI channel is
almost achieved.%
\begin{figure}[t]
\noindent \begin{centering}
\includegraphics[scale=0.49]{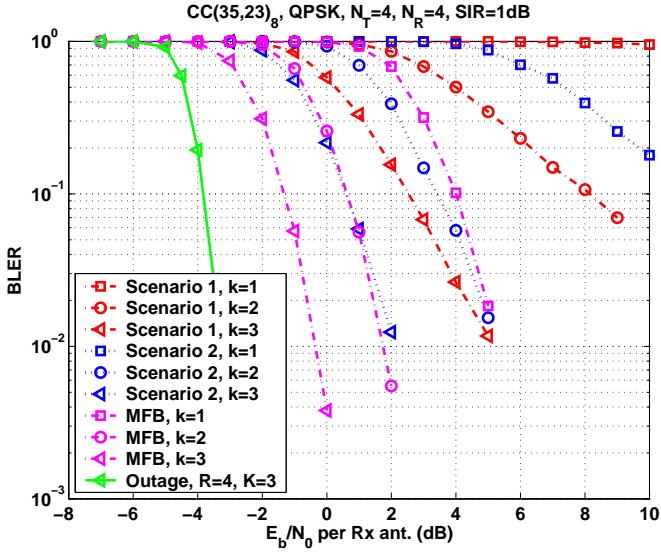}
\par\end{centering}

\caption{\label{fig:BLER_4x4_SIR1dB} BLER performance for CC $\left(35,23\right)_{8}$,
QPSK, $N_{T}=N_{R}=4$, $L=2$ equal energy paths, and $\mathrm{SIR}=1$dB.
Scenario 1: $N'_{T}=4$, $L'=4$, Scenario 2: $N'_{T}=2$, $L'=1$,
and $\rho_{k}=2,\, k=1,\cdots,K$. }

\end{figure}

\section{Conclusion\label{sec:Conclusion}}

In this paper, we investigated efficient iterative turbo packet combining
for broadband ST--BICM transmission with hybrid ARQ over CCI-limited
MIMO-ISI channels. We have introduced a frequency domain turbo combining
scheme where signals and CFRs corresponding to all ARQ rounds are
combined in a MMSE fashion to decode the data packet at each round.
The covariance of the overall (over all ARQ rounds) CCI plus noise
required by the frequency domain MMSE soft packet combiner is constructed
by separately computing the covariance related to each round. The
proposed technique has a complexity order cubic against the product
of the number of receive antennas and ARQ delay. This limitation is
overcome by an optimized recursive implementation algorithm where
complexity is only cubic in term of the number of transmit antennas.
We evaluated the computational load and memory requirements, and found
that the introduced recursive technique only requires few arithmetic
additions and memory compared to conventional LLR-level combining
schemes. We analyzed the effect of CCI channel rank on performance.
Interestingly, under a sum-rank condition, the frequency domain MMSE
soft packet combiner can completely remove CCI for asymptotically
high SNR. Finally, we provided simulation results where we showed
that the proposed technique achieves BLER performance superior to
LLR-level combining, and offers high CCI cancellation capability and
diversity order for many interference scenarios. 

\begin{center}
\textsc{Appendix}
\par\end{center}

\begin{center}
\textsc{Proof of Theorem \ref{thm:CCI_supress_prop}}
\par\end{center}

Under the assumption of perfect LLR feedback from the SISO decoder,
the frequency domain soft packet combiner output (\ref{eq:Forward_Backward_FD}),
at ARQ round $k$, can be expressed as,

\begin{equation}
\mathbf{z}_{f_{\mathrm{perfect\, LLR}}}^{\left(k\right)}=\mathbf{A}\mathbf{s}_{f}+\mathbf{x}_{f}^{\left(k\right)},\label{eq:FilerOutput_perfectLLRs}\end{equation}
where $\mathbf{A}$ is the diagonal matrix of frequency domain symbol
gains,\begin{multline}
\mathbf{A}=\mathrm{diag}\left\{ (\mathbf{G}_{0}^{\left(k\right)})_{1,1},\cdots,(\mathbf{G}_{0}^{\left(k\right)})_{N_{T},N_{T}},\cdots,\right.\,\,\,\,\,\,\,\,\,\,\,\,\,\,\,\,\\
\left.(\mathbf{G}_{T-1}^{\left(k\right)})_{1,1},\cdots,(\mathbf{G}_{T-1}^{\left(k\right)})_{N_{T},N_{T}}\right\} ,\label{eq:FilerOutput_perfectLLRs_MatrixA}\end{multline}
 with $\mathbf{G}_{i}^{\left(k\right)}=\underline{\mathbf{\Lambda}}_{i}^{\left(k\right)^{H}}\boldsymbol{{\Xi}}_{k}^{-1}\underline{\mathbf{\Lambda}}_{i}^{\left(k\right)}$,
and $\mathbf{x}_{f}^{\left(k\right)}$ is the filtered CCI plus thermal
noise at the output of the packet combining filter. Its covariance
matrix is \begin{equation}
\mathbf{G}^{\left(k\right)}=\mathrm{diag}\left\{ \mathbf{G}_{0}^{\left(k\right)},\cdots,\mathbf{G}_{T-1}^{\left(k\right)}\right\} .\label{eq:FilerOutput_perfectLLRs_MatrixG}\end{equation}
Now, let us examine the structure of matrix $\mathbf{G}_{i}^{\left(k\right)}$
for asymptotically high SNR, i.e., $\sigma^{2}\rightarrow0$. 

Let $\boldsymbol{{\Pi}}_{1}\boldsymbol{{\Pi}}_{1}^{H},\cdots,\boldsymbol{{\Pi}}_{k}\boldsymbol{{\Pi}}_{k}^{H}$
be the low-rank decompositions of matrices $\boldsymbol{{\Theta}}_{1}^{\mathrm{CCI}},\cdots,\boldsymbol{{\Theta}}_{k}^{\mathrm{CCI}}$,
where $\boldsymbol{{\Pi}}_{1}\in\mathbb{C}^{N_{R}\times\rho_{1}},\cdots,\boldsymbol{{\Pi}}_{k}\in\mathbb{C}^{N_{R}\times\rho_{k}}$.
For the sake of notation simplicity, we write $\sum_{u=1}^{k}\rho_{u}=\rho$.
It follows that the rank of $\boldsymbol{{\Pi}}=\mathrm{diag}\left\{ \boldsymbol{{\Pi}}_{1},\cdots,\boldsymbol{{\Pi}}_{k}\right\} $
is $\rho$, and $\boldsymbol{{\Xi}}_{k}=\boldsymbol{{\Pi}}\boldsymbol{{\Pi}}^{H}+\sigma^{2}\mathbf{I}_{kN_{R}}$
is a square invertible matrix. Therefore, it has an eigenvalue decomposition
(E.V.D) that can be expressed as,\begin{multline}
\boldsymbol{{\Xi}}_{k}=\underbrace{\left[\begin{array}{cc}
\mathbf{P}_{\rho} & \mathbf{P}_{kN_{R}-\rho}\end{array}\right]}_{\mathbf{P}}\left[\begin{array}{cc}
\boldsymbol{{\Upsilon}}+\sigma^{2}\mathbf{I}_{\rho}\\
 & \sigma^{2}\mathbf{I}_{kN_{R}-\rho}\end{array}\right]\times\,\,\,\,\,\,\,\,\,\,\,\,\,\,\,\,\\
\left[\begin{array}{c}
\mathbf{P}_{\rho}^{H}\\
\mathbf{P}_{kN_{R}-\rho}^{H}\end{array}\right],\label{eq:EVD}\end{multline}
where $\mathbf{P}\mathbf{P}^{H}=\mathbf{P}^{H}\mathbf{P}=\mathbf{I}_{kN_{R}}$
since $\boldsymbol{{\Xi}}_{k}$ is symmetric. This condition yields
the following set of equalities,

\addtocounter{equation}{0}\begin{subequations}

\begin{eqnarray}
\mathbf{P}_{\rho}^{H}\mathbf{P}_{\rho} & = & \mathbf{I}_{\rho},\label{eq:eigenmatrix_cond1}\\
\mathbf{P}_{kN_{R}-\rho}^{H}\mathbf{P}_{kN_{R}-\rho} & = & \mathbf{I}_{kN_{R}-\rho},\label{eq:eigenmatrix_cond2}\\
\mathbf{P}_{\rho}^{H}\mathbf{P}_{kN_{R}-\rho} & = & \boldsymbol{{0}},\label{eq:eigenmatrix_cond3}\\
\mathbf{P}_{\rho}\mathbf{P}_{\rho}^{H}+\mathbf{P}_{kN_{R}-\rho}\mathbf{P}_{kN_{R}-\rho}^{H} & = & \mathbf{I}_{kN_{R}}.\label{eq:eigenmatrix_cond4}\end{eqnarray}

\noindent \end{subequations}Therefore, a Taylor expansion of $\boldsymbol{{\Xi}}_{k}^{-1}$
when $\sigma^{2}\rightarrow0$, is given as, \begin{equation}
\boldsymbol{{\Xi}}_{k}^{-1}=\mathbf{P}_{\rho}\boldsymbol{{\Upsilon}}^{-1}\mathbf{P}_{\rho}^{H}+\sigma^{-2}\mathbf{I}_{kN_{R}}+\mathcal{O}\left(\sigma^{2}\right).\label{eq:TaylorExpansion}\end{equation}
Note that $\boldsymbol{{\Upsilon}}$ does not have any null diagonal
element, i.e., $\boldsymbol{{\Upsilon}}$ is invertible. Indeed, multiplying
the left and right sides of (\ref{eq:EVD}) by $\mathbf{P}^{H}$ and
$\mathbf{P}$, respectively, and with respect to (\ref{eq:eigenmatrix_cond1}),
we get, $\mathbf{P}_{\rho}^{H}\boldsymbol{{\Pi}}\boldsymbol{{\Pi}}^{H}\mathbf{P}_{\rho}=\boldsymbol{{\Upsilon}}$.
By noting that $\mathbf{P}_{\rho}^{H}\boldsymbol{{\Pi}}$ is $\rho\times\rho$
and has rank equal to $\rho$, it follows that $\boldsymbol{{\Upsilon}}^{-1}=\left(\boldsymbol{{\Pi}}^{H}\mathbf{P}_{\rho}\right)^{-1}\left(\mathbf{P}_{\rho}^{H}\boldsymbol{{\Pi}}\right)^{-1}$.
Therefore, when $\sigma^{2}\rightarrow0$, we have, \begin{equation}
\mathbf{G}_{i}^{\left(k\right)}=\underline{\mathbf{\Lambda}}_{i}^{\left(k\right)^{H}}\mathbf{P}_{\rho}\boldsymbol{{\Upsilon}}^{-1}\mathbf{P}_{\rho}^{H}\underline{\mathbf{\Lambda}}_{i}^{\left(k\right)}+\sigma^{-2}\underline{\mathbf{\Lambda}}_{i}^{\left(k\right)^{H}}\underline{\mathbf{\Lambda}}_{i}^{\left(k\right)}+\mathcal{O}\left(\sigma^{2}\right).\label{eq:MatrixG_i_kTaylorExpansion}\end{equation}
Since the time domain channel coefficients are i.i.d., it follows
that the $kN_{R}\times N_{T}$ matrix $\underline{\mathbf{\Lambda}}_{i}^{\left(k\right)}$
has full-column rank unless all fading coefficients are equal to zero.
If $\rho+N_{T}<kN_{R}$, i.e., $\rho<kN_{R}-N_{T}$, then all the
first $\rho$ columns of $\boldsymbol{{\Xi}}_{k}$ (column vectors
of $\mathbf{P}_{\rho}$) are in the kernel of $\underline{\mathbf{\Lambda}}_{i}^{\left(k\right)^{H}}$,
i.e., $\underline{\mathbf{\Lambda}}_{i}^{\left(k\right)^{H}}\mathbf{P}_{\rho}=\boldsymbol{{0}}_{N_{T}\times\rho}$.
It follows that, when $\sigma^{2}\rightarrow0$, \begin{equation}
\mathbf{G}_{i}^{\left(k\right)}=\sigma^{-2}\underline{\mathbf{\Lambda}}_{i}^{\left(k\right)^{H}}\underline{\mathbf{\Lambda}}_{i}^{\left(k\right)}+\mathcal{O}\left(\sigma^{2}\right).\label{eq:MatrixG_i_kTaylorExpansion_Exp2}\end{equation}
Therefore, when SNR$\rightarrow\infty$, we get

\begin{align}
\mathrm{SINR} & =\frac{1}{\sigma^{2}}T\sum_{i=0}^{T-1}\mathrm{tr}\left\{ \underline{\mathbf{\Lambda}}_{i}^{\left(k\right)^{H}}\underline{\mathbf{\Lambda}}_{i}^{\left(k\right)}\right\} +\mathcal{O}\left(\sigma^{2}\right)\nonumber \\
 & =\underbrace{\frac{1}{\sigma^{2}}\sum_{l=0}^{L-1}\sum_{u=1}^{k}\mathrm{tr}\left\{ \underline{\mathbf{H}}_{l}^{\left(u\right)^{H}}\underline{\mathbf{H}}_{l}^{\left(u\right)}\right\} }_{\mathrm{SNR_{MF}}}+\mathcal{O}\left(\sigma^{2}\right),\label{eq:Final_SINR_Expr}\end{align}

\noindent where $\mathrm{SNR_{MF}}$ corresponds to the instantaneous
matched filter (MF) SNR in the case of $k$ rounds CCI-free MIMO-ISI
ARQ channel.\hfill{} $\blacksquare$

\begin{center}
\textsc{Acknowledgment}
\par\end{center}

The authors would like to thank Dr. Matthew Valenti for coordinating
the review process, and the three anonymous reviewers for their very
helpful comments and suggestions. 

\vspace{-1cm}
\begin{biography}[{\includegraphics[width=1in,clip,keepaspectratio]{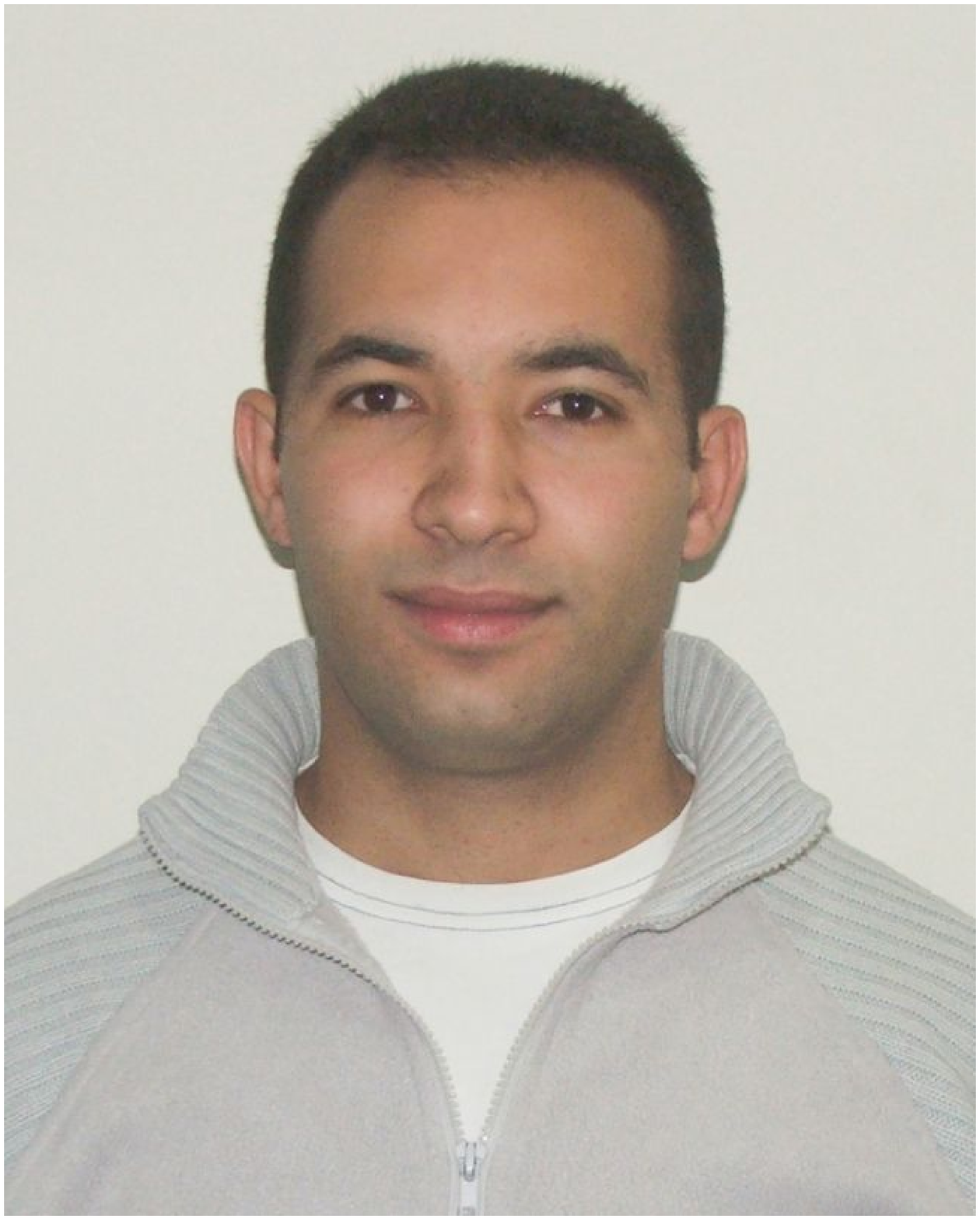}}]{Tarik Ait-Idir}
(S'06-M'07) was born in Rabat, Morocco, in 1978. He received the "diplôme d'ingénieur d'état" in telecommunications from INPT, Rabat, and the Ph.D. degree in electrical engineering from Telecom Bretagne, Brest, France, in 2001, and 2006, respectively. He is currently an Assistant Professor of wireless communications at the Communication Systems department, INPT, Rabat. He is also an adjunct researcher with Institut Telecom / Telecom Bretegne/LabSticc. From July 2001 to February 2003 he was with Ericsson. His research interests include PHY and cross-layer aspects of MIMO systems, relay communications, and dynamic spectrum management. Dr. Ait-Idir has been on the technical program committee of several IEEE conferences, including ICC, WCNC, PIMRC, and VTC, and chaired some of their sessions. He has been a technical co-chair of the MIMO Systems Symposium at IWCMC 2009, and IWCMC 2010.   
\end{biography}

~

~

~

\noindent \begin{biography}[{\includegraphics[width=1in,clip,keepaspectratio]{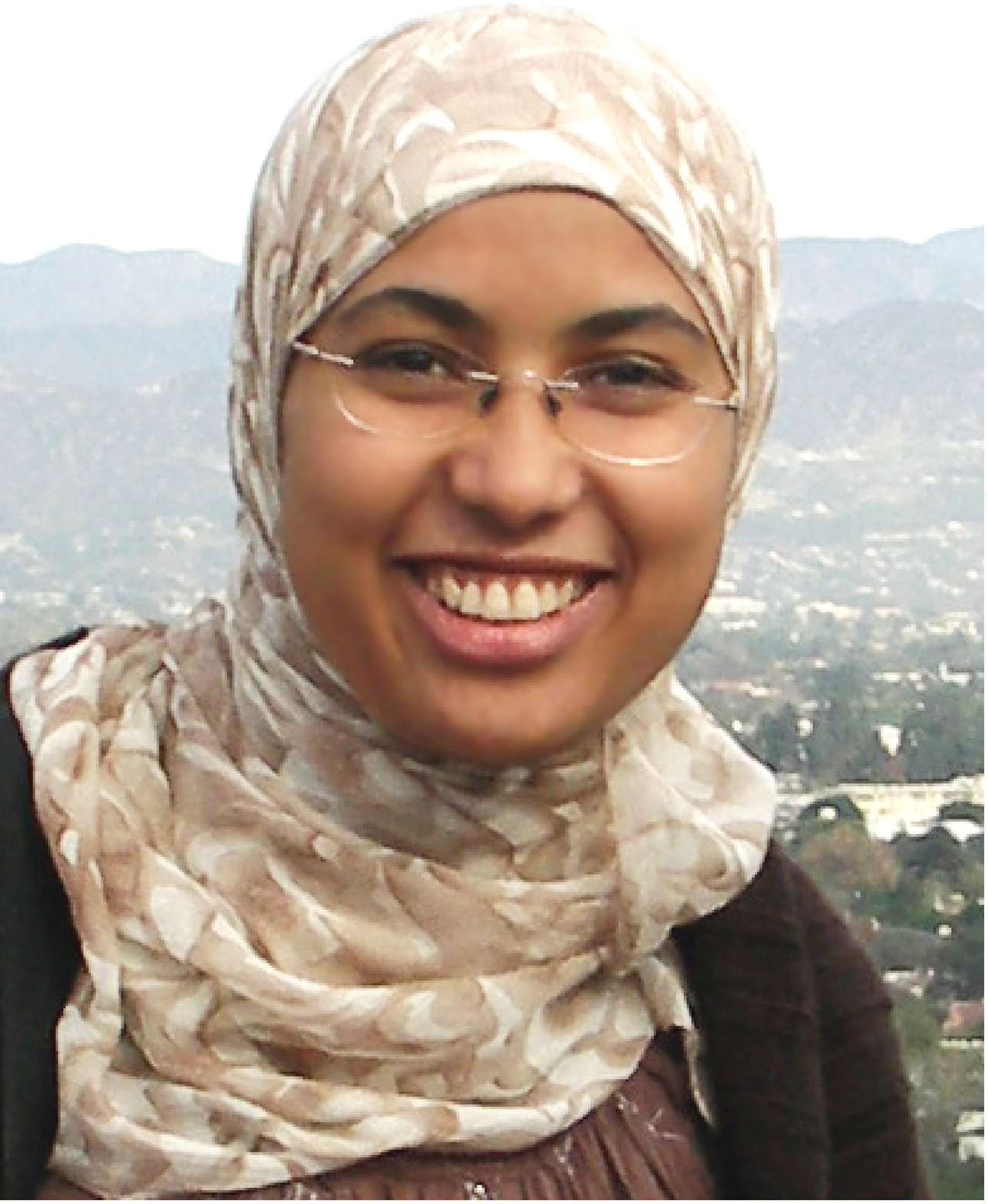}}]{Houda Chafnaji}
received the "diplôme d'ingénieur d'état" in telecommunications in 2004 from INPT, Rabat, Morocco, and  the M.Sc. in telecommunications in 2006 from INRS, Montreal, Canada. She is currently a PhD. candidate at Telecom-Bretagne, Brest, France. Since 2006, she has been a lecturer at the department of Communications Systems, INPT. Her research interests include wireless communications with a focus on physical and MAC layers design, Hybrid ARQ, cooperative communication, packet combining, and performance evaluation of wireless communications systems.   
\end{biography}

~

~

~

\noindent \begin{biography}[{\includegraphics[width=1in,clip,keepaspectratio]{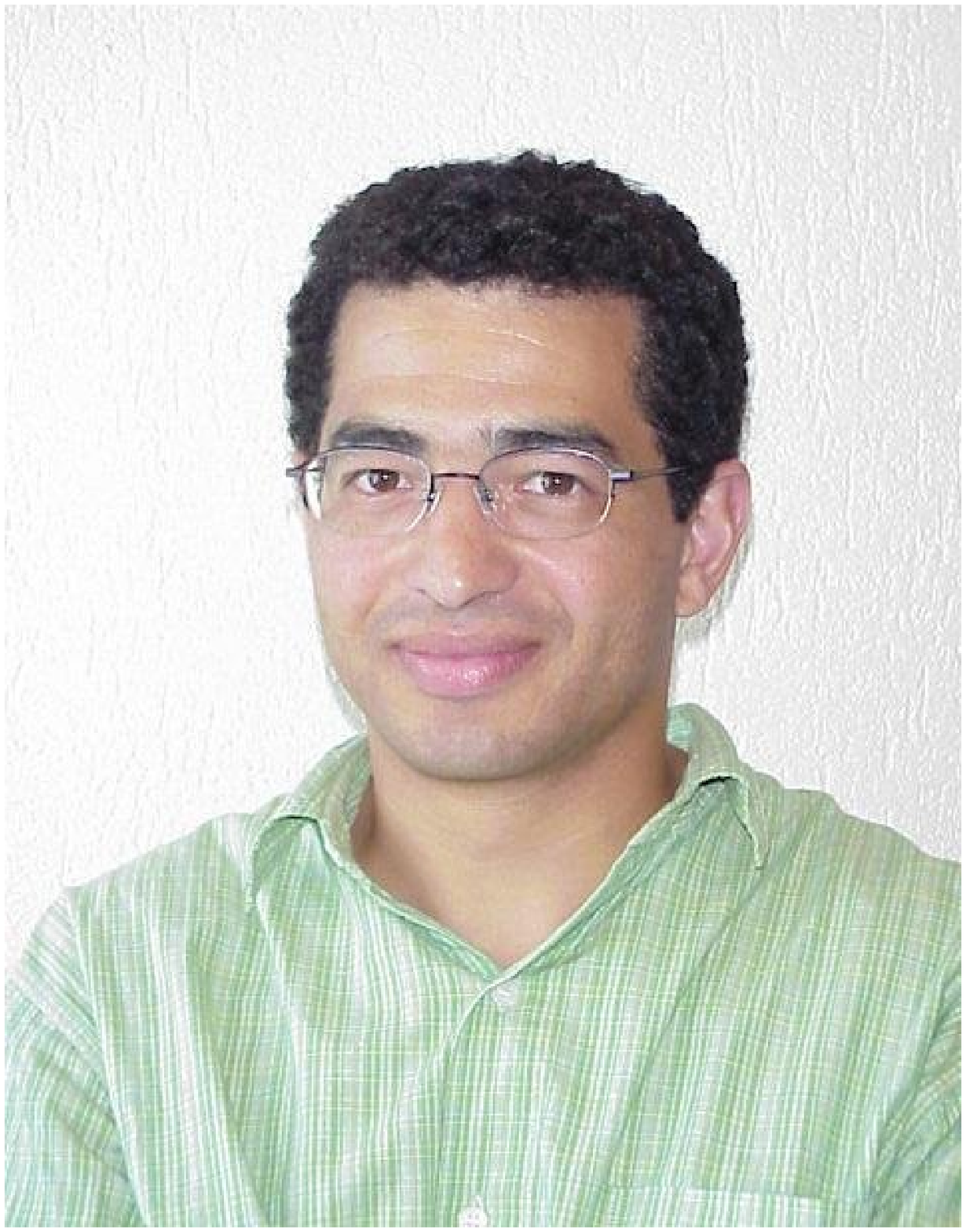}}]{Samir Saoudi}
(M'01-SM'09) was born in Rabat, Morocco, on November 28, 1963. He received the "diplôme d'ingénieur d'état" from ENST Bretagne, Brest, France, in 1987, the Ph.D. degree in telecommunications from the 'Université de Rennes-I' in 1990, and the "Habilitation à Diriger des Recherches en Sciences" in 1997. Since 1991, he has been with the Signal and Communications department, Institut Telecom / Telecom Bretegne/LabSticc, where he is currently a Professor. He is also with Université Européenne de Bretagne. In summer 2009, he has visited Orange Labs-Tokyo. His research interests include speech and audio coding, non parametric probability density function estimation, CDMA techniques, multiuser detection and MIMO techniques for UMTS and HSPA applications. His teaching interests are signal processing, probability, stochastic processes and speech processing. Dr. Saoudi supervised more than 20 Ph.D. Students. He is the author and/or coauthor of around eighty publications. He has been the general chairman of the second International Symposium on Image/Video Communications over fixed and mobile networks (ISIVC'04), and technical co-chair of the MIMO systems symposium at IWCMC 2010.
\end{biography}
\end{document}